%% file: main_arxiv.tex
\documentclass[sigconf,review=false,nonacm]{acmart}

\usepackage{graphicx} 
\usepackage{amsmath}
\usepackage{subcaption}
\usepackage{hyperref}
\usepackage{fancyhdr}
\usepackage[linesnumbered,ruled,vlined]{algorithm2e}

\renewcommand\footnotetextcopyrightpermission[1]{} 
\author{Jiayao Wang}
\affiliation{%
 \institution{National University of Defense Technology}
 \city{Changsha}
 \country{China}}

\author{Qilong Shi}
\affiliation{%
 \institution{Tsinghua University}
 \city{Beijing}
 \country{China}}

 \author{Xiyan Liang}
\affiliation{%
 \institution{Nankai University}
 \city{Tianjin}
 \country{China}}

 \author{Han Wang}
\affiliation{%
 \institution{Peking University}
 \city{Beijing}
 \country{China}}

\author{Wenjun Li}
\affiliation{%
 \institution{Peng Cheng Laboratory}
 \city{Shenzhen}
 \country{China}}

\author{Ziling Wei}
\affiliation{%
 \institution{National University of Defense Technology}
 \city{Changsha}
 \country{China}}

 \author{Weizhe Zhang}
\affiliation{%
 \institution{Harbin Institute of Technology}
 \city{Harbin}
 \country{China}}

\author{Shuhui Chen}
\affiliation{%
 \institution{National University of Defense Technology}
 \city{Changsha}
 \country{China}}

\begin{document}

\title{PSSketch: Finding Persistent and Sparse Flow with High Accuracy and Efficiency}

\begin{abstract}
   Finding persistent sparse (PS) flow is critical to early warning of many threats. Previous works have predominantly focused on either heavy or persistent flows, with limited attention given to PS flows. Although some recent studies pay attention to PS flows, they struggle to establish an objective criterion due to insufficient data-driven observations, resulting in reduced accuracy. In this paper, we define a new criterion ``anomaly boundary'' to distinguish PS flows from regular flows. Specifically, a flow whose persistence exceeds a threshold will be protected, while a protected flow with a density lower than a threshold is reported as a PS flow. We then introduce PSSketch, a high-precision layered sketch to find PS flows. PSSketch employs variable-length bitwise counters, where the first layer tracks the frequency and persistence of all flows, and the second layer protects potential PS flows and records overflow counts from the first layer. Some optimizations have also been implemented to reduce memory consumption further and improve accuracy. The experiments show that PSSketch reduces memory consumption by an order of magnitude compared to the strawman solution combined with existing work. Compared with SOTA solutions for finding PS flows, it outperforms up to 2.94x in F1 score and reduces ARE by 1-2 orders of magnitude. Meanwhile, PSSketch achieves a higher throughput than these solutions.
\end{abstract}

\keywords{Data streams, Data mining, Approximate algorithm, Sketch}

\maketitle

\footnotetext{* The source code of this paper can be downloaded from the GitHub~\cite{github}.}

\section{Introduction}
    
    The increasing data flow poses a growing challenge in ensuring network security through traffic analysis. Approximate flow processing algorithms have gained popularity due to the excessive resource consumption associated with the exact analysis of all flows. Sketch is widely used as a typical data compression tool. In the existing literature, the task is often divided into spatial and temporal perspectives, where the number of occurrences and time span will be recorded. Heavy flow detection is a common task. By approximating flows' occurrence times, it finds flows that appear significantly more than others \cite{CMS,AS,OMS,SFS,Cold,DCF,PS,SBF,DHS}. The time span is typically represented by dividing the time into equal parts, called time windows, and counting the number of windows in which a flow has appeared, known as flow persistence. Persistent flow detection has also received increasing attention in recent years\cite{KF,BF,OOS,PIS}.

    However, many attacks are not heavy flows; they tend to hide themselves and persist for a long time. For example, backdoor programs or reverse proxies may send a few packets continuously to the target to ensure stable control or retrieve information. Beyond attacks, detecting PS flows is also valuable for server administrators. Some users continuously access over a long period and consume many server resources but contribute little. For instance, in online games, account sellers may create a large number of low-activity accounts to get awards every day.

    Current frequency-focused algorithms focus primarily on heavy flows, making them incapable of capturing infrequent flows. Therefore, they cannot be applied in the above scenarios. However, existing algorithms designed to track persistent flows do not account for frequency statistics. An intuitive strawman solution is to combine the two kinds of approach. However, as shown in our experiments, without optimization for the PS task, such a solution exhibits low memory efficiency, severely limiting its practicality. PISketch\cite{PIS} is currently the only model optimized for the PS task. However, their criterion used to filter PS flows is not based on data observations. Thus, the filtered flows are not always persistent and sparse. Moreover, they utilize a simple traditional sketch to store flow information, leading to low accuracy.

    In this paper, we introduce a novel model, PSSketch, with the aim of finding PS flows with high accuracy and efficiency. To achieve that, our design revolves around the following points:

    \begin{itemize}

    \item Objective criterion. We analyze three commonly used datasets from real networks. The persistence and density distribution of flows exhibit similar cliff-like characteristics. Thus, we can draw an ``anomaly boundary'', using persistence and density to divide PS flows.

    \item Protecting PS Flows. Traditional sketches store all flows in a single structure, where new flows often replace target flows (i.e., PS flows in our task) due to memory limitation. We propose a layered data structure, PSSketch, in which the first layer stores persistent flows while the second layer stores PS candidates.

    \item Bit-Level Counters and Overflow Storage. Existing two-layer sketch structures often suffer from data redundancy, where information is stored twice. To address this issue, the first layer contains only small bit-wise counters while the second layer counts its overflow times, significantly improving memory utilization.

    \item Optimization. We raise a probability replacement to protect PS flows in their early stages and a pruning mechanism to eliminate non-PS flow to better use memory. Single Instruction Multiple Data (SIMD) is also used to reduce memory access, thereby improving throughput noticeably.

    \end{itemize}

    Experiments show that, compared to the strawman solution that combines traditional sketch structures, PSSketch achieves a reduction in memory requirements by 1 to 2 orders of magnitude. Compared to the SOTA PS flow finding model, PISketch, we increase the F1 Score from 0.4-0.6 to over 0.99. Additionally, the average relative error in information storage is reduced by two orders of magnitude. Furthermore, PSSketch significantly outperforms both the strawman solution and PISketch in terms of throughput.

    In Section 2, we will introduce the tasks of traffic analysis and related work. The observation and criterion we used are detailed in Section 3. Then, we introduce the proposed model PSSketch in Section 4 and mathematical analysis in Section 5. After that, PSSketch is comprehensively evaluated in Section 6. Finally, Section 7 concludes our work.

\begin{figure*}[htbp]
    \centering
    \begin{subfigure}[b]{0.26\textwidth}
        \includegraphics[width=\textwidth]{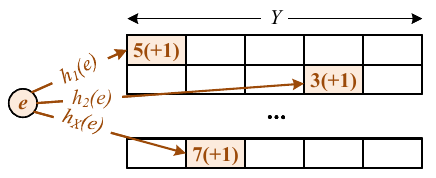}
        \caption{CMSketch}
        \label{CMS}
    \end{subfigure}
    \begin{subfigure}[b]{0.33\textwidth}
        \includegraphics[width=\textwidth]{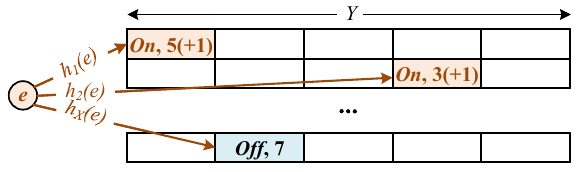}
        \caption{On-off Sketch}
        \label{OOS}
    \end{subfigure}
    \begin{subfigure}[b]{0.29\textwidth}
        \includegraphics[width=\textwidth]{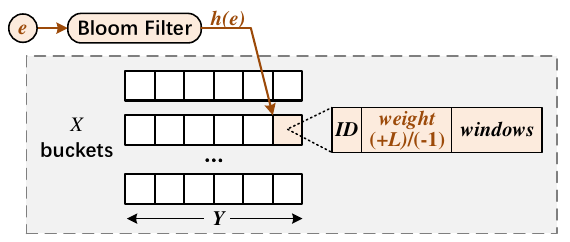}
        \caption{PISketch}
        \label{PIS}
    \end{subfigure}
    \vspace{-0.15cm}
    \caption{Structure of related works.}
    \vspace{-0.3cm}
    \label{SOTA}
\end{figure*}

\section{Background and Related Work}

    \subsection{Background}   

    Before introducing our task, we first give some key definitions; detailed versions can be found in the Appendix \ref{app_extra}.

    \textbf{Flow:} Packets containing the same 5-tuple (SrcIP, DstIP, SrcPt, DstPt, Proto) are treated as coming from the same flow $e$.

    \textbf{Frequency:} The total number of packets a flow contains is denoted as its frequency $f_e$.

    \textbf{Time Window:} We divide a longer time interval into many small intervals of equal length $t$, each of which is called a time window $T_i$.

    \textbf{Persistence:} The number of time windows covered by packets from a certain flow $e$ is denoted as its persistence $p_i$. It will only be incremented by one or kept constant in each time window.

    \textbf{Density:} On longer time spans, the density can be calculated directly from the frequency $f_e$ and persistence $p_e$ as $d_e = \frac{f_e}{p_e}$.

     Unlike packet classification\cite{PC1,PC2,PC3,PC4,PC5} and Deep Packet Inspection\cite{DPI1,DPI2,DPI3,DPI4,DPI5}, finding specific data flows requires storing statistical information, such as frequency and persistence. However, at the end of the 20th century, the rapidly increasing amount of data made whole storage impossible. To address this issue, probabilistic data structures such as Filters\cite{BF,CF,XF,CBF} and Sketches\cite{CMS,AS,OMS,SFS,Cold,DCF,PS,SBF,DHS,BM,CC,Chain,HG,HK,MV,UAS,TS} are introduced. They use Hash functions to store the information in a small space, saving memory consumption by sacrificing accuracy. The trade-off between accuracy and memory consumption is the key to this issue. In Section 2.2, we will introduce how related works count frequency and persistence and lead to a new task proposed in recent years: finding flows with both high persistence and low frequency. It means that we need to take statistics into consideration while filtering out low-frequency flows, in contrast to finding heavy flows.
    
    Some other approaches \cite{Shift,AF,ES,FS,UM,EE,CB,RC,RS,IBF} have also been raised for finding specific flows; however, due to their limited adoption, they will not be discussed in this paper.

    \subsection{Related Works}
    
    \textbf{Frequency Estimation.} Network measurement is currently divided into two main categories: frequency estimation and persistence estimation. Frequency estimation refers to counting the number of occurrences for each flow. Count-Min Sketch (CMSketch) \cite{CMS} is one of the most widely applied models. To illustrate the core concept of frequency estimation, we use CMSketch as an example.

    As shown in Figure \ref{SOTA}(a), CMSketch consists of $X$ buckets, each with $Y$ counters. When a flow arrives, its ID is processed by $X$ different hash functions, generating $X$ indices $I_1,I_2,...,I_x$. The corresponding counters at each index are then incremented. To query the frequency of a flow, the same process is performed, and it returns the minimum value among the $X$ counters. 
    The key challenge of all sketch-based solutions is balancing accuracy and performance. Beyond that, due to memory constraints, most Sketch-based frequency estimation solutions only focus on identifying ``heavy flows''. They target the few flows that occur with top frequency. Thus, they make the newly arriving flows replace those with low frequencies. So far, none of these models focuses on counting infrequent flows, which makes it impossible for them to find PS flows.

    \textbf{Persistence Estimation.} Kary Filter\cite{KF} is the first to incorporate the temporal dimension by comparing the frequency of flows across different time intervals to detect changes. Burst Sketch\cite{BS} further extends this research within the Filter-Sketch framework. On-off Sketch (OOSketch) \cite{OOS} explicitly introduced the concept of time windows, using persistence to determine whether a flow consistently appears over an extended period. 
    It introduces an on-off switch based on CMSketch so that the same counter can only be modified once in a time window as shown in Figure \ref{SOTA}(b). These solutions do not account for the number of occurrences within each window, leaving them unsuitable for finding PS flows.

    \textbf{Finding PS flows.} In our application scenario, PS flows are defined by their low frequency and high persistence. A straightforward approach is to combine existing models. We propose a \textit{Strawman solution}, which uses CMSketch to track frequency and OOSketch to record persistence. However, it suffers from low memory efficiency and throughput. PISketch\cite{PIS} is currently the only model specifically optimized for PS flows, as it integrates frequency and persistence. As illustrated in Figure \ref{SOTA}(c), PISketch employs a filter to check whether a flow has appeared in the current time window and records a weight ($W$) in the sketch. If a flow appears for the first time in a time window, its weight is increased by a predefined value $L$ ($L>1$). Each subsequent occurrence of the flow within the same time window decreases the weight by one, and the flow is removed once $W$ reaches zero. Although the weight indicates PS flow characteristics to some extent, it does not serve as an accurate criterion, as we will introduce in Section 3. Moreover, the reliance on a relatively simple sketch structure in PISketch frequently results in interference or even the replacement of potential PS flows by other flows.

    \section{Observation and criterion}

    An intuitive fact is that a flow must be persistent to be considered an anomaly. A few occasional packets are not statistically sufficient to identify any behavior. Therefore, we first require the persistence of a flow to reach a certain threshold. Subsequently, we filter the lower-density flows to report the PS flows. To determine the thresholds for persistence and density, we study three common real-net datasets: CAIDA\cite{CAIDA}, MAWI\cite{MAWI}, and Campus\cite{PS}. They all exhibit the same distribution characteristics. Figure \ref{DIST} takes CAIDA as an example; most flows have persistence below 10, while flows with persistence above 50 account for only 3.141\%. We set the $p_0$ threshold within the range of 20 to 60. A larger persistence threshold leads to a more strict criterion of PS flow. Regarding density, most flows exhibit values below 2, and the number of flows decreases as density increases. Conversely, as density approaches 1, the number of flows drops sharply, indicating that regular flows with high persistence rarely have densities near 1. Therefore, we set the \(d_0\) threshold between 1.1 and 1.5. A smaller density threshold indicates more rigorous searching. Distribution characteristics for the other datasets are provided in Appendix \ref{app_extra}.
    
\begin{figure}[h!]
    \centering
    \begin{subfigure}[b]{0.22\textwidth}
        \includegraphics[width=\textwidth]{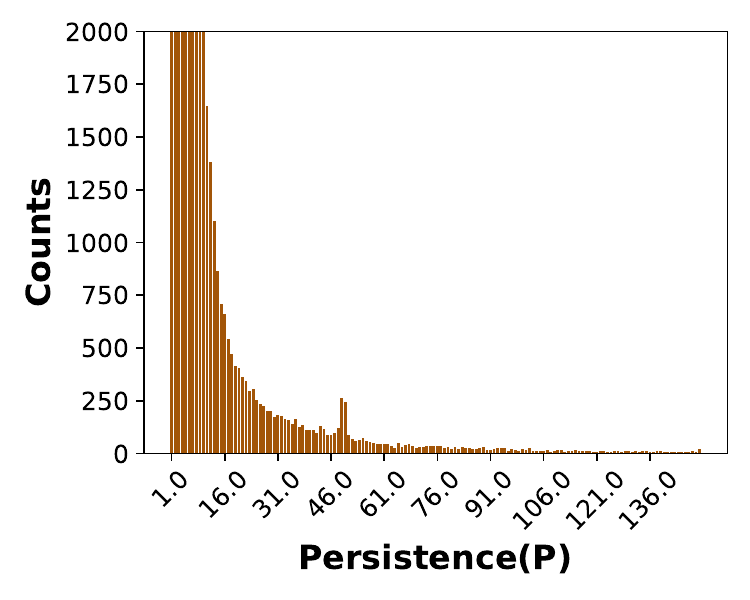}
        \caption{Persistence}
        \label{p_caida}
    \end{subfigure}
    \begin{subfigure}[b]{0.22\textwidth}
        \includegraphics[width=\textwidth]{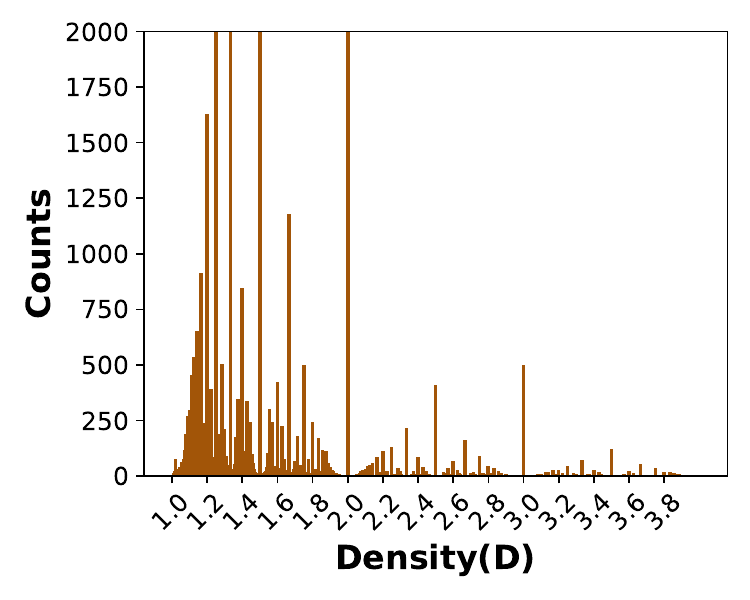}
        \caption{Density}
        \label{p_campus}
    \end{subfigure}
    \vspace{-0.15cm}
    \caption{The distribution characteristics of CAIDA.}
    \label{DIST}
\end{figure}
    
    

    Figure \ref{PS_define} demonstrates the advantage of using density over PISketch using weights $W$. We report PS flows that are much more concentrated near the $l:f=p$ line. This means that we can preferentially find those sparse flows, while the flows found by PSSketch are distributed over the entire possible range without discrimination. Advantages under other datasets can be found in the Appendix \ref{app_extra}.

    We define the ``anomalous boundary'' as follows: a flow is reported as a PS flow if its persistence exceeds the threshold \(p_0\) (\(p_0 \in [20, 60]\)) and then its density falls below the threshold \(d_0\) (\(d_0 \in [1.1, 1.5]\)). A more lenient boundary reduces the false negatives in threat warning but increases the workload. In practical deployment, the thresholds should be adjusted according to performance and security requirements. Based on that, we will introduce our data structures and algorithms in Section 4.

    \begin{figure}[h!]
    \centering
    \begin{subfigure}[b]{0.22\textwidth}
        \includegraphics[width=\textwidth]{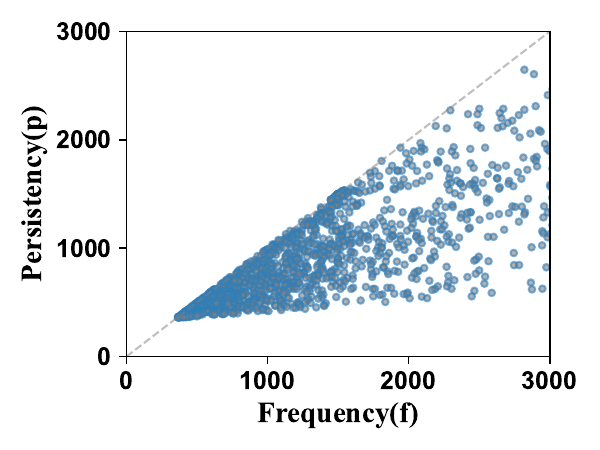}
        \caption{PISketch}
        \label{PISketch}
    \end{subfigure}
   \begin{subfigure}[b]{0.22\textwidth}
        \includegraphics[width=\textwidth]{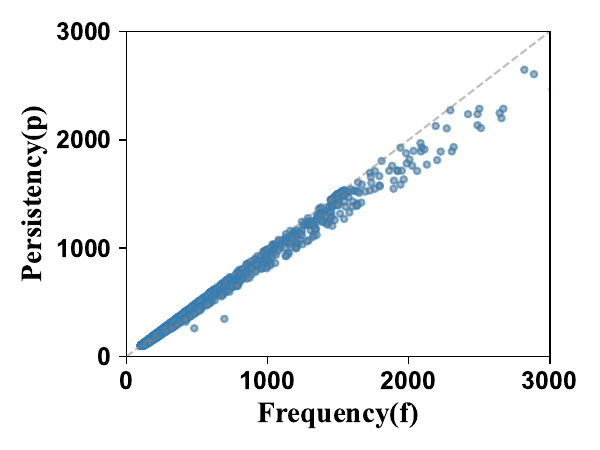}
        \caption{PSSketch}
        \label{PSSketch}
    \end{subfigure}
    \vspace{-0.15cm}
    \caption{Top-2K PS flows reported under the Campus dataset.}
    \label{PS_define}
\end{figure}

\setlength{\fboxsep}{1pt} 
    
\section{Proposed methods}

    \subsection{Data Structure}
    
    As illustrated in Figure \ref{str}, PSSketch is composed of two layers: the Competition Layer (CL) and the Protection Layer (PL). The Competition Layer consists of $X$ buckets. Each bucket contains $Y$ entries. One entry is responsible for storing information about one flow, including the flow's fingerprint ($FP$), frequency ($f$), persistence ($p$), and flags for storing additional information. Thus, we can denote the $j^{th}$ entry in $i^{th}$ bucket as $Bk[i].Et[j]=\{FP,f,p,flags\}$. The length of $FP$, counter $f$, and counter $p$ can vary depending on the load, while the $flags$ is fixed at 2 bits, including a window identifier $W$ and an overflow identifier $OF$. The flag $W$ is initialized to 0 at the beginning of each time window and is set to 1 when the flow in the entry comes in that window. The flag $OF$ indicates whether the flow has reported an overflow. We denote the two flags as $\fbox{01}$. All flows enter the CL first. With the need of storing a large amount of information, the counters are compressed to under 8 bits. Once the counter of a flow in CL overflows, it is reported to PL. The Protection Layer consists of an entry vector of length $R$. The content of an entry includes a flow's ID and the overflow times of its CL counters, that is, $Prt[i]=\{ID_e,f_{of}^e,p_{of}^e\}$. Since the persistent flows represent only a small portion of the total flows, the Protection Layer can store $ID$, which is longer than $FP$, to reduce hash collisions. 
        
    Recording overflow times greatly reduces the length of counters in two layers, avoiding information redundancy. In addition, the information transmission is unidirectional from CL to PL, saving the delay introduced by the information exchange. However, the above two problems are the important bottlenecks of many traditional double-layer sketches.

   \begin{figure}[H]
        \centering
            \includegraphics[width=\linewidth]{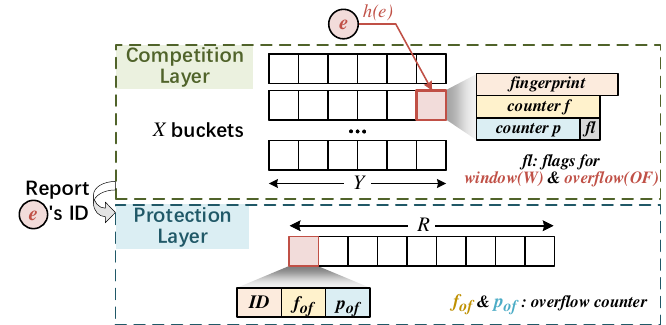}   
        \caption{Structure of PSSketch.}
        \label{str}
    \end{figure}

\subsection{Operations}

    \subsubsection{Insert} First, we apply a hash function to a flow's ID and get its fingerprint $FP_e = h(e)$. The $m^{th}$ bucket is selected where $m = FP_e \mod X$. In that bucket, there are three possible cases:

    \textit{Case 1:} Search $BK[m]$ and find an entry that $BK[m].Et[n].FP=FP_e$. If it is found, it transitions to $Update$.

    \textit{Case 2:} If no such entry, we search for the first empty entry where $BK[m].Et[n].FP = 0$ and fill it with ${FP=FP_e, f=1, p=1, flag= \fbox{00}}$. Therefore, $e$ is added to $BK[m].Et[n]$ successfully.

    \textit{Case 3:} If none of the entries in $BK[m]$ satisfy $BK[m].Et[n].FP = FP_e$ and none are empty, $BK[m]$ is full and transitions to $Contend$.

    The pseudocode of $Insert$ can be found in Appendix \ref{app_extra}.
 
    \subsubsection{Update} When a flow $e$ has already been recorded in $BK[m].Et[n]$, we inspect the flag $W$. If $BK[m].Et[n].flag[W] == 0$, flow $e$ has not appeared in the current time window. Thus, $BK[m].Et[n].p$ is incremented by 1. Moreover,  $BK[m].Et[n].f$ is incremented by 1 and $BK[m].Et[n].flag[W]$ is updated to 1. Note that we will reset all $flag[W]$ to 0 at the beginning of each time window.

    Since the counters used in the Competition Layer are small, $f$ and $p$ may soon overflow after incremented. The two cases are:

    \textit{Case 1:} $f$ overflows first. In our design, the width of counter $f$ is 8 bits, while counter $p$ is 6 bits. If $f$ overflows first, it indicates that $f_e$ has reached $2^8 = 256$, while $p_e$ is at most $2^6-1 = 63$. In this case, the density $D = f/p > 4$. Flow $e$ will not be reported in this case because even if it meets the condition to become a persistent flow (i.e., $p_e > p_0$), its high density disqualifies it from being protected. Instead, we eliminate the flow.

    \textit{Case 2:} $p$ overflows first. At this point, flow $e$ meets the condition for persistence, and its density $D = f/p < 4$. The flow will be reported to the PL for protection. The $Flag[OF]$ is now set to 1, and $BK[m].Et[n].p$ is reset to 0.

    \textit{Case 3:} If the $Flag[OF]$ has already been 1, this flow isprotected. Whether $f$ or $p$ overflows, the overflow is directly reported to the PL, and both $BK[m].Et[n].f$ and $BK[m].Et[n].p$ are reset.

    Based on cases above, the PL performs two types of updates:
    
    \textit{Item Creation (Case 2):} When flow $e$ enters the PL for the first time, it searches for the first empty entry, denoted as $Prt[k]$. It writes $\{ID_e, 0, 1\}$ into it, representing the flow's identifier, counter $f$ overflow times, and counter $p$ overflow times. If there is no empty entry, We replace the flow with the largest density. Since persistent flows account for a small fraction of all flows, we can avoid this by appropriately increasing the length of PL according to the load.

    \textit{Item Update (Case 3):} If flow $e$ has already been protected, no matter $f$ or $p$ overflows, the Protection Layer finds the entry corresponding to the flow (also denoted as $Prt[k]$ here) and increments either $f$ or $p$ accordingly. Then, the $report$ returns "SUCCESS."

    The pseudocode of $Update$ can also be found in Appendix \ref{app_extra}.

\subsubsection{Contend} Contend occurs when the bucket $BK[m]$ is filled in CL. The newly inserted flow $e$ will attempt to replace an existing flow in the bucket. Since, both the frequency $f$ and persistence $p$ of the new flow $e$ are equal to 1, its density cannot be considered. Thus, we identify the entry with the smallest persistence value, denoted as $Et[h]$. To maximize the protection of already stored flows, we overwrite $BK[m].Et[h]$ with a probability of $1/BK[m].Et[h].p$. Note that a protected flow (i.e., $Flag[OF]$ is set to 1) cannot be replaced. Such an approach implies that a flow with a large persistence value will be hard to replace. In our preliminary experiments, we find that all flows evicted have a persistence below 4, with most at 1 or 2. It strengthens the protection of persistent flows in their early stages. \textit{Contend} will bring ``Ejection Errors'', the main source of error for sketch-based solutions, which we analyze in Section 5.

    \subsubsection{Query} A query can be executed at any time. The competition Layer will be traversed to find each position where $FP > 0$. If $Flag[W]==1$ in that entry, the flow is persistent. Then, we locate its $ID$ in the Protection Layer to retrieve the overflow counts. Denote $V_{PL}$ as the counter value in the Protection Layer and $V_{CL}$ as the counter value in the Competition Layer. The final value $V$ can be calculated as $V = V_{prt} \times 2^{L_{c}} + V_{c}$, where $L_{c}$ denotes the bit-length of the counter in the Competition Layer.

    In this paper, we utilize an 8-bit counter for frequency ($f$) and a 6-bit counter for persistence ($p$) in the Competition Layer. Suppose their values are $f = 5$ and $p = 18$, while the corresponding Protection Layer values are $f_{of} = 1$ for frequency and $p_{of} = 3$ for persistence. The frequency of the flow, $f_e$, can be computed as: $f_e = 1 \times 2^8 + 5 = 261$. Similarly, the persistence, $p_e$, is calculated as: $p_e = 3 \times 2^6 + 18 = 210$. Thus, the density of the flow $d_e$ is given by: $d_e = \frac{f_e}{p_e} = 1.243$. Any flow with a density below the threshold $d_0$ is reported as a PS flow, while other flows are considered persistent.

    \subsection{Optimization}

    \subsubsection{Prune} In our observations for the majority cases, if a flow is reported as a PS flow, it will not have a sudden increase in frequency throughout the entire timeline. Also, in our application scenario, a latent attack should not generate a high volume of packets at any time. Therefore, if an entry in PL contains an overflow counter $f_{of}=\mu$ larger than its $p_{of}=\tau$, the density will be $d_e > \frac{2^{8}\mu}{2^{6}(\tau+1) - 1} = \frac{256\mu}{64(\tau+1) - 1} > 4   (\mu > \tau, \; \mu, \tau \in \mathbb{Z}^+)$. As a result, the flow can be cleared from both CL and PL even though it is being protected.
    
    \subsubsection{Burst Elimination} Interestingly, we have found several flows generate a high volume of traffic in a short period, but exhibit PS flow characteristics during all other time intervals. With their relatively high overall density, these flows may not meet the PS flow criterion in the entire period. With the aim to report these flows for security, if the overflow counter for $f$ increases more than twice in one time window, we only add 2 to the counter $f_{of}$.
    
    \subsubsection{One-time Traversal} During \textit{Insert}, we need to traverse $BK[m]$ up to three times in order to search for $FP_e$, empty slot, and the flow with the smallest persistence. In practical use, we introduced three extra counters at the CL for each bucket (not each entry): an empty slot indicator $Ep[m]$, a replacement indicator $Rp[m]$, and a replacement counter $MinP[m]$, all initialized to -1, as shown in Figure \ref{OTT}. They will help us find all three locations in one loop.

       \begin{figure}[H]
\setlength{\abovecaptionskip}{0.1cm}
        \centering
        \includegraphics[width=0.9\linewidth]{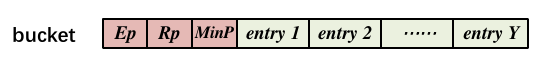}
        \vspace{-0.1cm}
        \caption{Three extra counters for each bucket.}
        \label{OTT}
        \vspace{-0.3cm}
\end{figure}

\subsection{A Running Example}
    Here, we present a simplified example of PSSketch operations. The green part represents one of the buckets in the Competition Layer with four entries, while the blue part represents the Protection Layer. The width of the counter $f$ is 8 bits, thus its overflow value is 256. In contrast, the overflow value of counter $p$ is set to the persistence threshold of 50 in our example. The following describes how we handle coming flows:

    \textit{Example for Competition Layer.} As shown in Figure \ref{cl}, when $e_1$ comes for the first time and is hashed to this bucket, it searches an empty entry and records $\langle FP_{e_1},1,1,\fbox{10} \rangle$. Then flow \(e_1\) arrives for the second time within this time window; its counter $f$ is incremented by one while its counter $p$ remains unchanged. Subsequently, a new flow \(e_3\) hashes to the same bucket; however, due to the lack of empty entries, $contend$ is initiated against \(e_1\), which has the lowest value of counter $p$. Since \(p_{e_1}=1\), the replacement probability is $1/p_{e_1} = 1$, resulting in \(e_1\) being replaced by \(e_3\).

    \begin{figure}[htbp]
    \setlength{\abovecaptionskip}{0.2cm}
        \centering
        \includegraphics[width=1\linewidth]{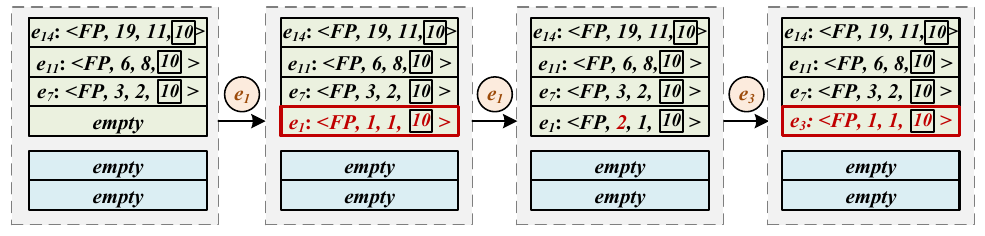}
        \caption{An Example of $Insert$, $Update$ and $Contend$.}
        \label{cl}
        \vspace{-0.35cm}
\end{figure}

    \textit{Example for Reporting to Protection Layer.} After the last update of flow $e_2$, as shown in Figure \ref{pl}, it reaches the overflow value of counter $p$. It is, therefore reported to the Protection Layer with initial information $\langle ID_{e_2},0,1 \rangle$. On the other hand, flow \(e_5\) comes and makes the counter $f$ to be 256. Thus, it reports an overflow of counter $f$. Upon reporting, its information in the Protection Layer updates to $\langle FP_{e_5},4,3\rangle$, satisfying the pruning condition \(f_{of} > p_{of}\). As a result, its information is cleared from both layers.

        \begin{figure}[H]
    \setlength{\abovecaptionskip}{0.2cm}
        \centering
        \includegraphics[width=1.02\linewidth]{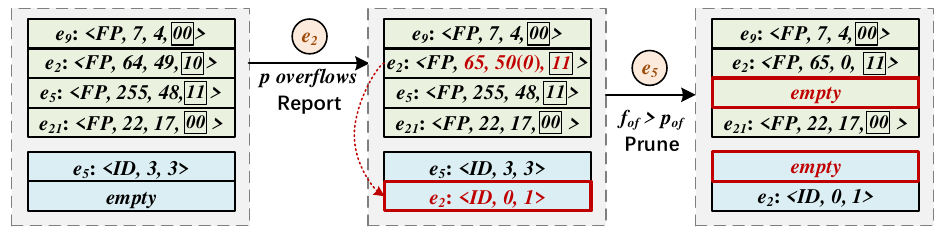}
        \caption{An Example of $Report$ and $Prune$.}
        \label{pl}
        \vspace{-0.25cm}
    \end{figure}

\section{Mathematical Analysis}
\input{Mathematical_Analysis}

\begin{figure*}[tbp]
    \centering
    \begin{subfigure}[b]{0.4\textwidth}
        \includegraphics[width=1\textwidth]{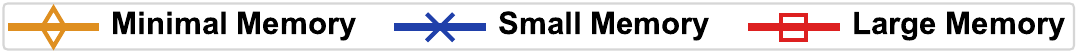}
        \caption*{}
        \vspace{-0.5cm}
        \label{}
    \end{subfigure}

    \begin{subfigure}[b]{0.163\textwidth}
        \includegraphics[width=\textwidth]{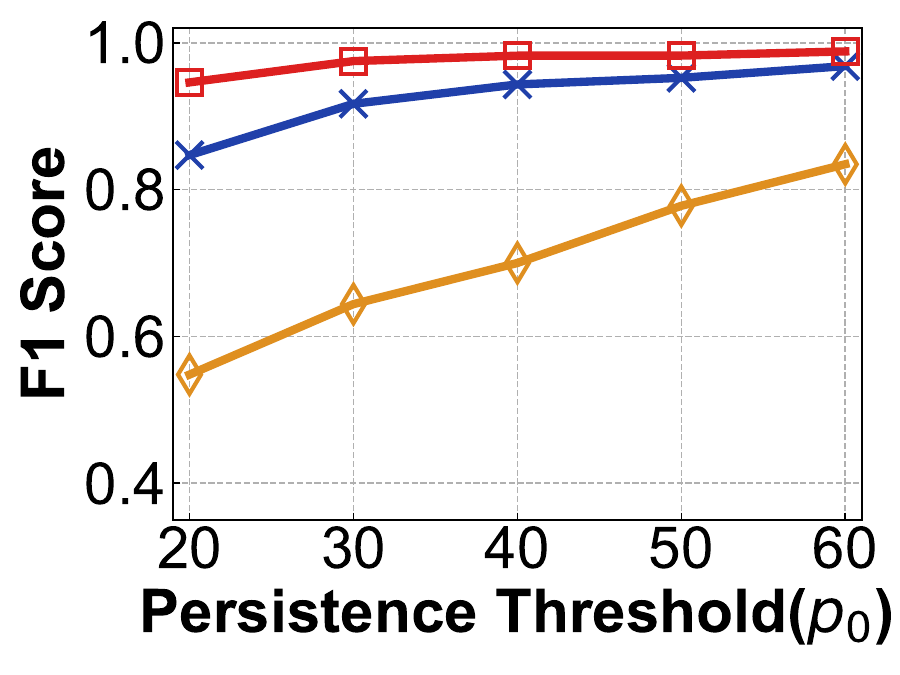}
        \caption{$p_0$-F1,CAIDA}
        \label{p_caida_f1}
    \end{subfigure}
    \begin{subfigure}[b]{0.163\textwidth}
        \includegraphics[width=\textwidth]{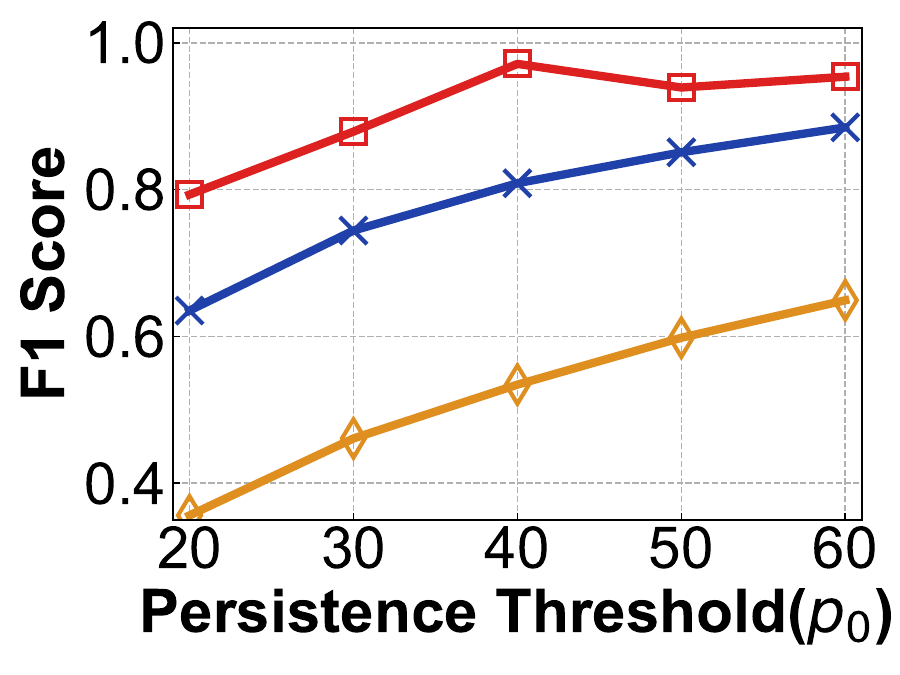}
        \caption{$p_0$-F1,Campus}
        \label{p_campus_f1}
    \end{subfigure}
    \begin{subfigure}[b]{0.163\textwidth}
        \includegraphics[width=\textwidth]{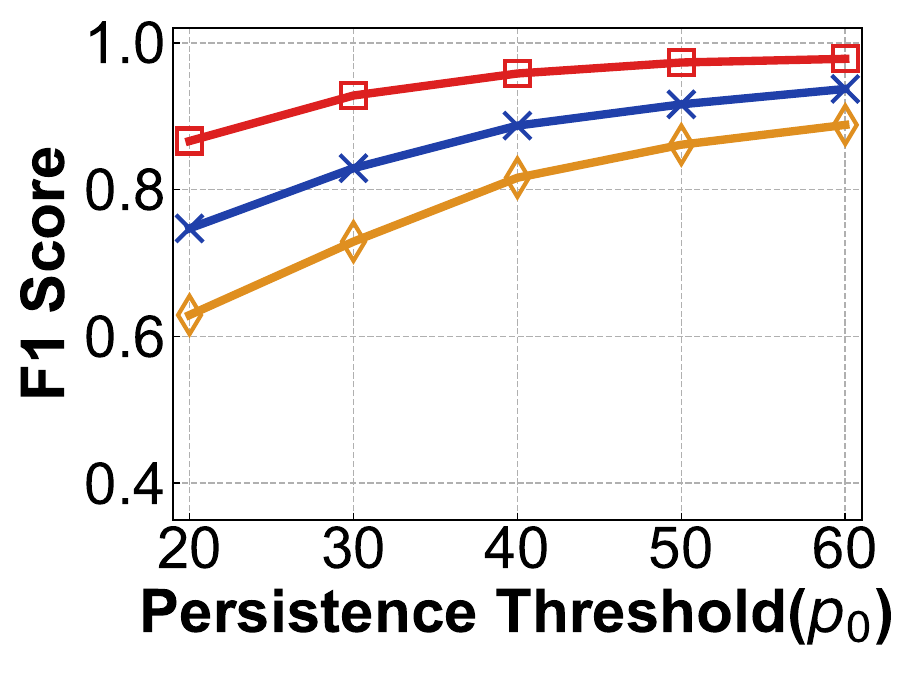}
        \caption{$p_0$-F1,MAWI}
        \label{p_mawi_f1}
    \end{subfigure}   
    \begin{subfigure}[b]{0.163\textwidth}
        \includegraphics[width=\textwidth]{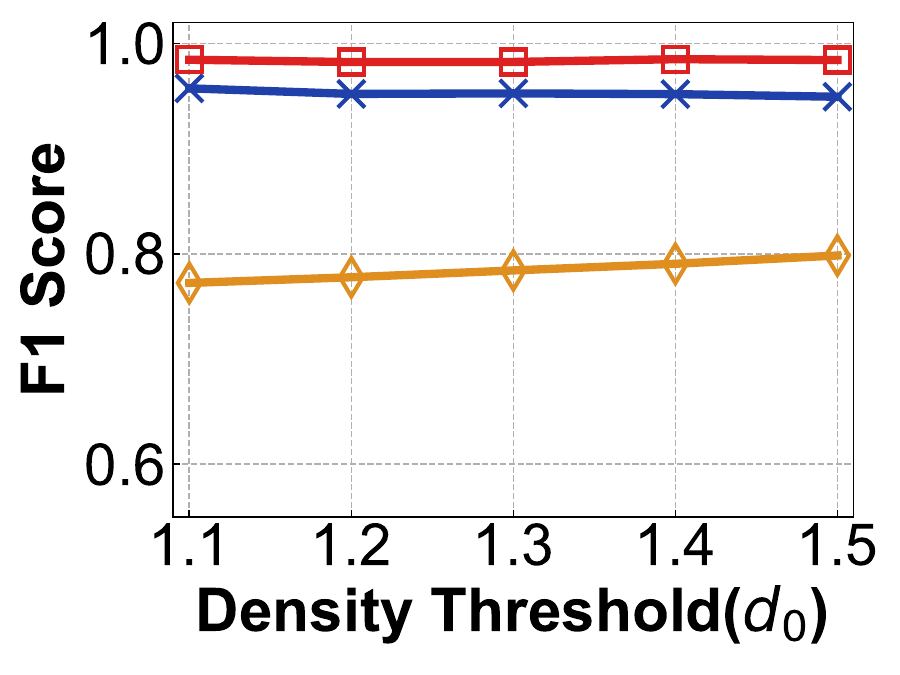}
        \caption{$d_0$-F1,CAIDA}
        \label{d_caida_f1}
    \end{subfigure}
    \begin{subfigure}[b]{0.163\textwidth}
        \includegraphics[width=\textwidth]{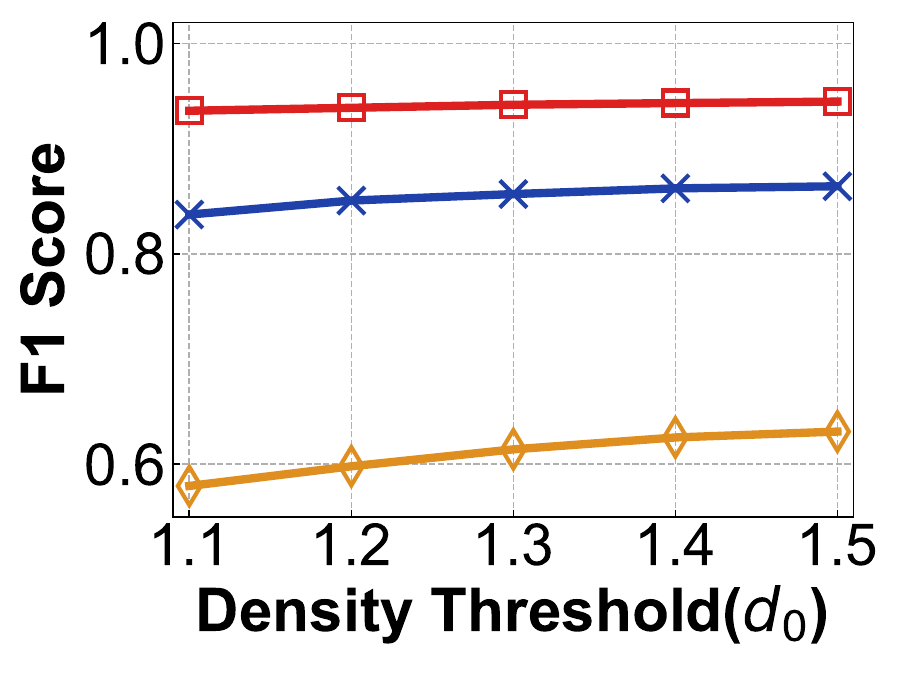}
        \caption{$d_0$-F1,Campus}
        \label{d_campus_f1}
    \end{subfigure}
    \begin{subfigure}[b]{0.163\textwidth}
        \includegraphics[width=\textwidth]{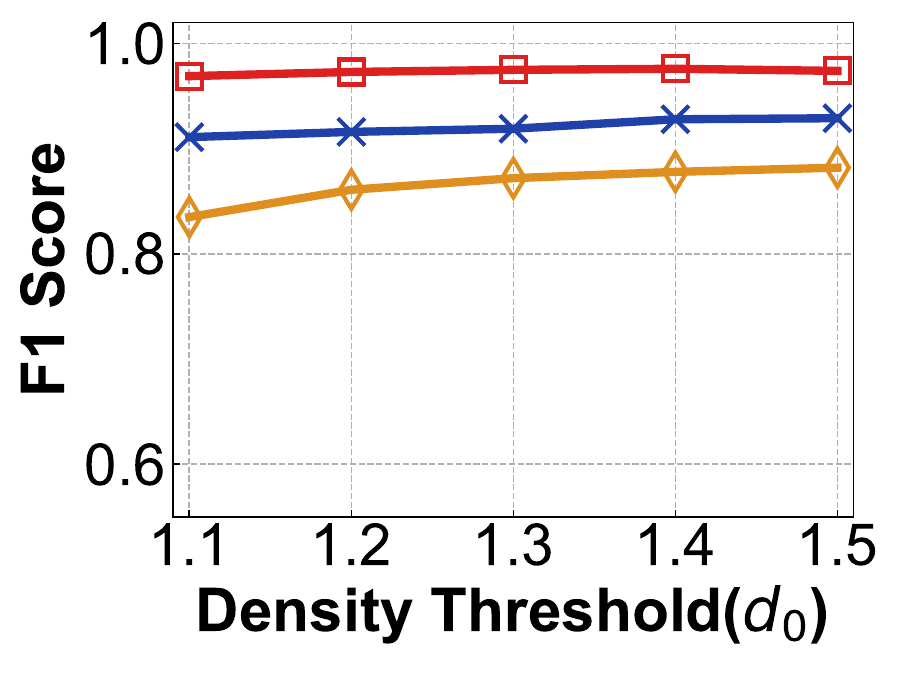}
        \caption{$d_0$-F1,MAWI}
        \label{d_mawi_f1}
    \end{subfigure}
    \vspace{-0.35cm}
    \caption{The impact of persistence threshold $p_0$ and density threshold $d_0$ on F1 Score.}
    \label{exp_pd_f1}
    
\end{figure*}


\begin{figure*}[tbp]
    \centering
    \begin{subfigure}[b]{0.4\textwidth}
        \includegraphics[width=1\textwidth]{Leg.pdf}
        \caption*{}
        \vspace{-0.5cm}
        \label{}
    \end{subfigure}
    
    \begin{subfigure}[b]{0.163\textwidth}
        \includegraphics[width=\textwidth]{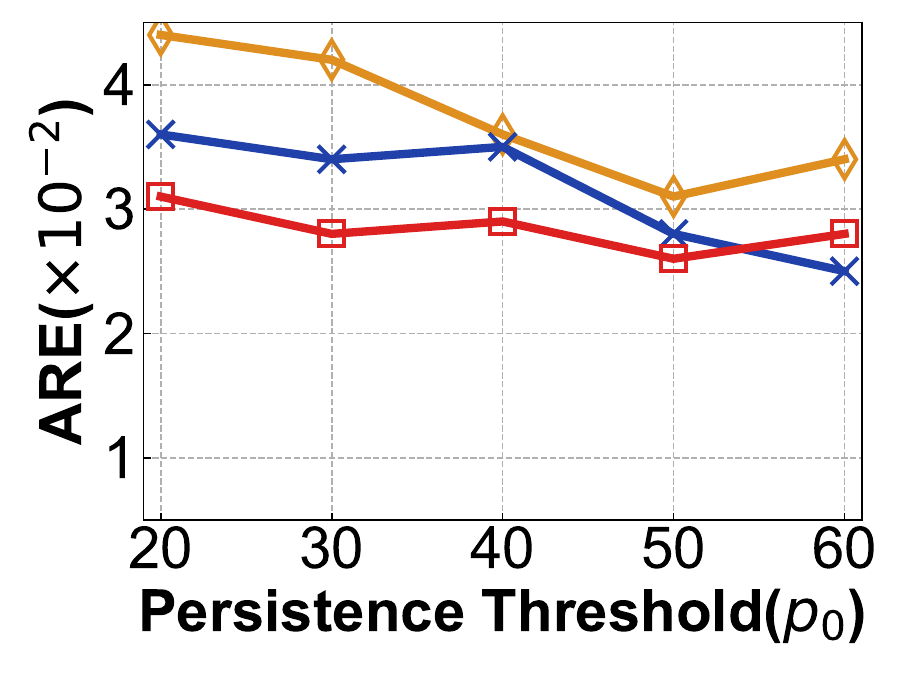}
        \caption{$p_0$-ARE,CAIDA}
        \label{p_caida_are}
    \end{subfigure}
    \begin{subfigure}[b]{0.163\textwidth}
        \includegraphics[width=\textwidth]{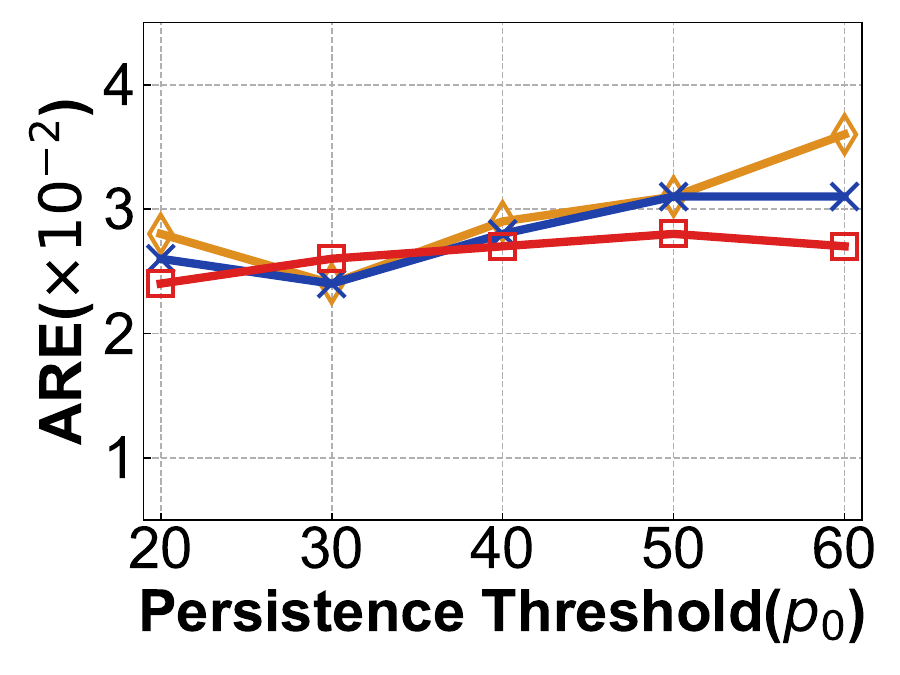}
        \caption{$p_0$-ARE,Campus}
        \label{p_campus_are}
    \end{subfigure}
    \begin{subfigure}[b]{0.163\textwidth}
        \includegraphics[width=\textwidth]{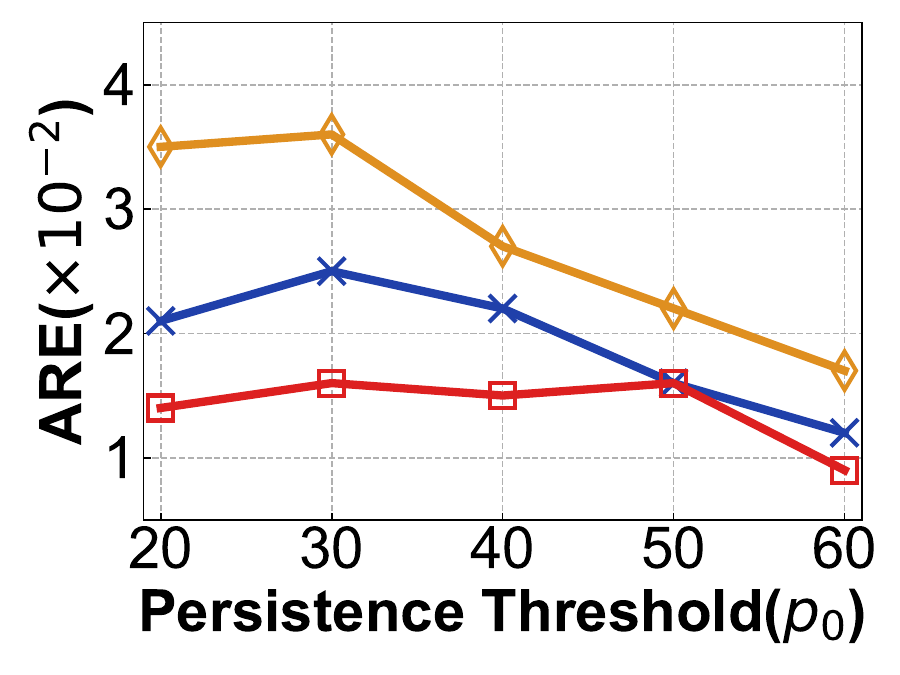}
        \caption{$p_0$-ARE,MAWI}
        \label{p_mawi_are}
    \end{subfigure}   
    \begin{subfigure}[b]{0.163\textwidth}
        \includegraphics[width=\textwidth]{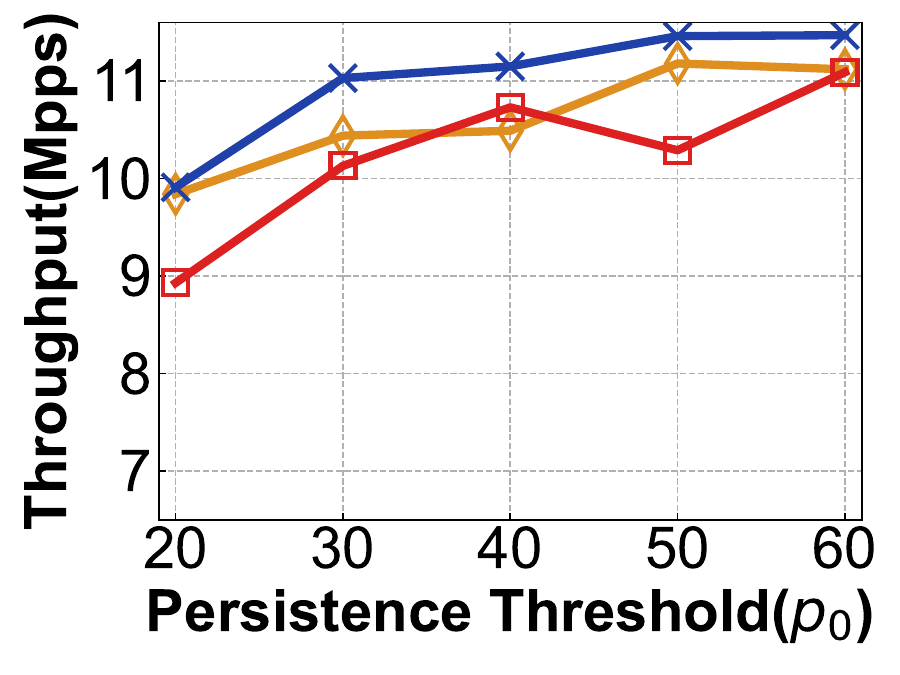}
        \caption{$p_0$-Tp,CAIDA}
        \label{p_caida_tp}
    \end{subfigure}
    \begin{subfigure}[b]{0.163\textwidth}
        \includegraphics[width=\textwidth]{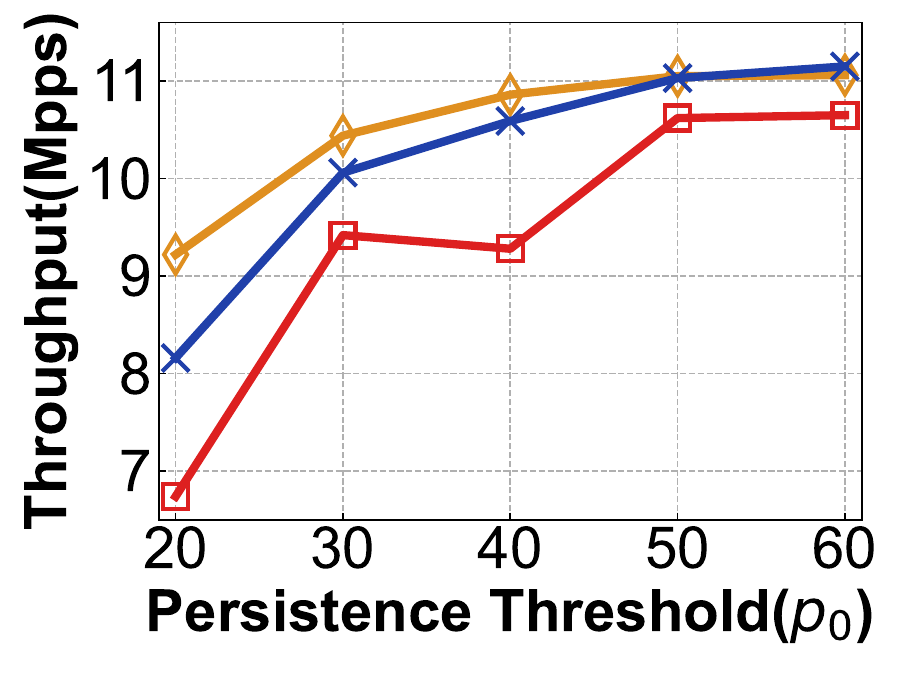}
        \caption{$p_0$-Tp,Campus}
        \label{p_campus_tp}
    \end{subfigure}
    \begin{subfigure}[b]{0.163\textwidth}
        \includegraphics[width=\textwidth]{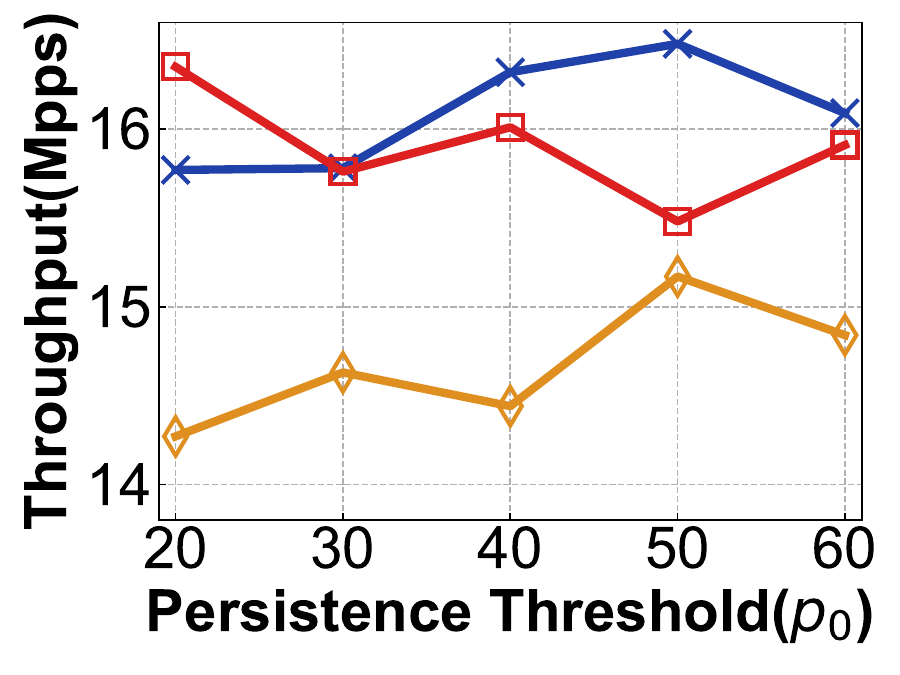}
        \caption{$p_0$-Tp,MAWI}
        \label{p_mawi_tp}
    \end{subfigure}   
    \vspace{-0.35cm}
    \caption{The impact of persistence threshold $p_0$ on ARE and Throughput.}
    \label{exp_pd_are}
\end{figure*}

\section{Evaluation}
    
    \subsection{Setup}

    We implement PSSketch in C++ with Bob Hash \cite{BOB}. The experiments are run on a PC with AMD R9 7940H, 16 Cores, and 32GB DRAM. We have released our source code at the GitHub~\cite{github}. The datasets we used are as follows:

    \textit{CAIDA.} It is a set of anonymous IP collected in 2018 \cite{CAIDA} which contains 2,490K packets with 109,534 flows. The proportion of PS flow is 1.055\%.

    \textit{Campus.} This dataset is collected from a real network within a campus, utilizing the same data as \cite{PS}. It contains 10,000K packets, comprising 259,948 distinct flows, with PS flows accounting for 1.433\% of the total flows.

    \textit{MAWI.} A real traffic trace provided by MAWI Working Group\cite{MAWI}. It contains 2,000K packets with 200,471 flows. PS flows only take 0.132\% total.

    \subsection{Comparing Solutions and Metrics}

    In this study, we will compare PSSketch with three distinct models:

    \textit{Strawman Solution:} As introduced in Section 2.2, we utilize CMSketch to estimate frequency and On-off Sketch to estimate persistence, while using an extra array to store information of PS flows.

    \textit{PISketch:} PISketch\cite{PIS} represents the current SOTA in detecting PS flow specially. PSSketch employs a traditional Sketch data structure and defines the concept of ``weight'' as a criterion for PS flows.

    \textit{PISketch-Density:} This model is based on PISketch, with modifications made to its criterion. It now selects PS flows by our proposed definition, ``density'' instead of ``weight'', while all other processing procedures remain unchanged. By comparing this model, we aim to illustrate the structural optimizations introduced by PSSketch.

    By comparing the PS flows reported by different solutions with the actual PS flows in the datasets, we will derive the following two metrics through comparison with the answer set:

    \textbf{F1 Score:} The calculation of the F1 Score is given by the formula $F1 = \frac{2 \cdot Precision \cdot Recall}{Precision + Recall}$, where $Precision$ means the proportion of PS flows in the predicted set and $Recall$ shows the proportion of PS flows successfully predicted in the answer set. 

    \textbf{ARE:} The calculation of the Average Relative Error (ARE) is defined as $ARE = \frac{1}{n} \sum_{i=1}^{n} \left| \frac{y_i - \hat{y}_i}{y_i} \right|$, where $\hat{y}_i$ is the value stored in the sketch and $y_i$ is its actual value. $n$ shows the total flows stored at the end. 

    \textbf{Throughput:} we will also evaluate the throughput of the three models, defined as the number of packets processed per second.

    \subsection{Parameter Evaluation}

    Our model incorporates three variables: the persistence threshold \(p_0\), the density threshold \(d_0\), and the aspect ratio of the Competition Layer (\(X:Y\)). This experiment is mainly used to test the sensitivity and robustness of PSSketch in different scenarios and workloads. Considering the different loads of different datasets, under the CAIDA and Campus datasets, we set minimal memory to 50 KB, small memory to 100 KB, and large memory to 150 KB. Under MAWI, the three values correspond to 15KB, 25KB, and 50KB, respectively.

    The impact of the persistent threshold \( p_0 \) and the density threshold \( d_0 \) on the F1 Score under varying memory constraints $Mem$ is illustrated in Figure \ref{exp_pd_f1}. The width $Y$ is now set to 32, while $X=\frac{Mem}{32}$. As \( p_0 \) increases, the threshold for entering the Protection Layer rises, resulting in a significant enhancement of the F1 Score, particularly pronounced under memory-constrained conditions. This observation suggests that PSSketch performs better in more stringent threat detection scenarios. In contrast, the impact of the threshold \( d_0 \) on the F1 Score is relatively modest, yet it exhibits a positive correlation.

    Figure \ref{exp_pd_are}(a), \ref{exp_pd_are}(b) and \ref{exp_pd_are}(c) illustrate the influence of \( p_0 \) on the average relative error. The optimal threshold \( p_0 \) varies according to load intensity; for higher intensities, such as in the Campus dataset, the optimal \( p_0 \) is lower. Conversely, as the load intensity decreases, the optimal threshold \( p_0 \) gradually increases.

    Figure \ref{exp_pd_are}(d),\ref{exp_pd_are}(e) and \ref{exp_pd_are}(f) depict the effect of \( p_0 \) on throughput. In the CAIDA and Campus datasets, stricter criteria significantly enhance the throughput of PSSketch, especially under conditions of ample memory. However, in the MAWI dataset, characterized by the lowest load intensity, the improvements in throughput are not obvious. Overall, the findings indicate that PSSketch achieves higher F1 Scores and throughput when addressing more rigorous threat detection tasks.

    Note that in the experiment mentioned above, we do not present the impact of the threshold \( D \) on ARE and throughput. This is because \( D \) primarily influences \textit{query}, not $insert$. Consequently, the effect of \( D \) on these two metrics is minimal.

    We designate the number of entries $Y$ within each bucket as the independent variable, thereby establishing $X=\frac{Mem}{Y}$. Figure \ref{exp_xy} shows the impact of aspect ratio, where the X-axis represents $y$, the factor of Y, which satisfies $Y=2^{y+1}$ for visual presentation. When we increase the bucket width $Y$ while reducing the number of buckets $X=\frac{Mem}{Y}$, the F1 Score initially rises and then declines. Interestingly, under varying memory constraints, the optimal bucket width for all datasets consistently falls within the 16-32 range, with a notable concentration around 16. From the overall trend of ARE, it is evident that increasing the capacity of each bucket, rather than simply increasing the number of buckets, is more effective in reducing the average relative error. Moreover, increasing the number of entries per bucket leads to more memory access per iteration. As a result, it is reasonable that throughput decreases as the bucket width \( Y \) increases.

\begin{figure}[tbp]
    \centering
    \begin{subfigure}[b]{0.36\textwidth}
        \includegraphics[width=1\textwidth]{Leg.pdf}
        \caption*{}
        \vspace{-0.5cm}
        \label{}
    \end{subfigure}
    
    \begin{subfigure}[b]{0.156\textwidth}
        \includegraphics[width=\textwidth]{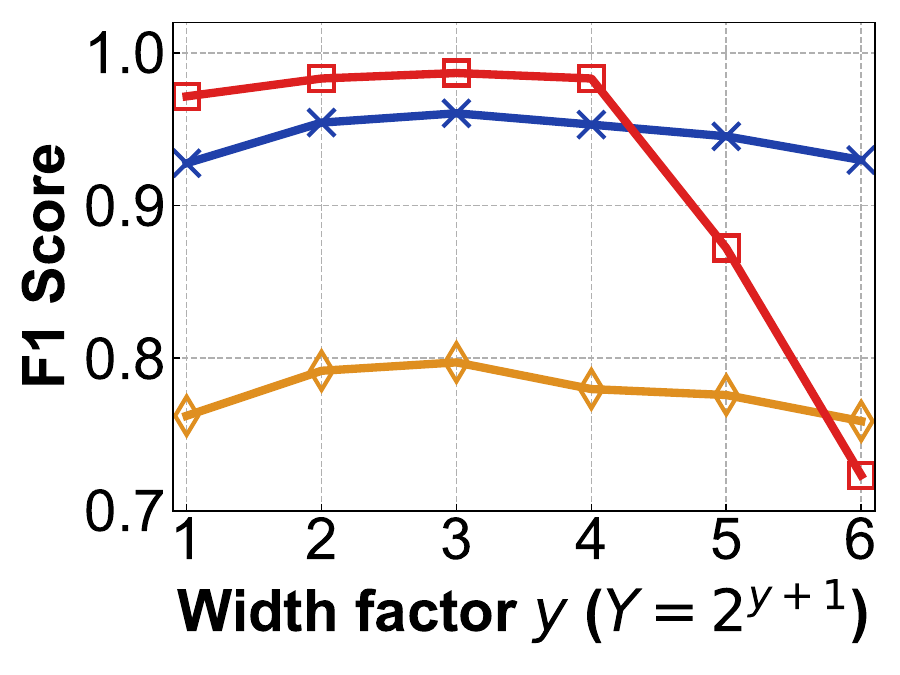}
        \caption{$Y$-F1,CAIDA}
        \label{xy_caida_f1}
        \vspace{0.25cm}
    \end{subfigure}
    \begin{subfigure}[b]{0.156\textwidth}
        \includegraphics[width=\textwidth]{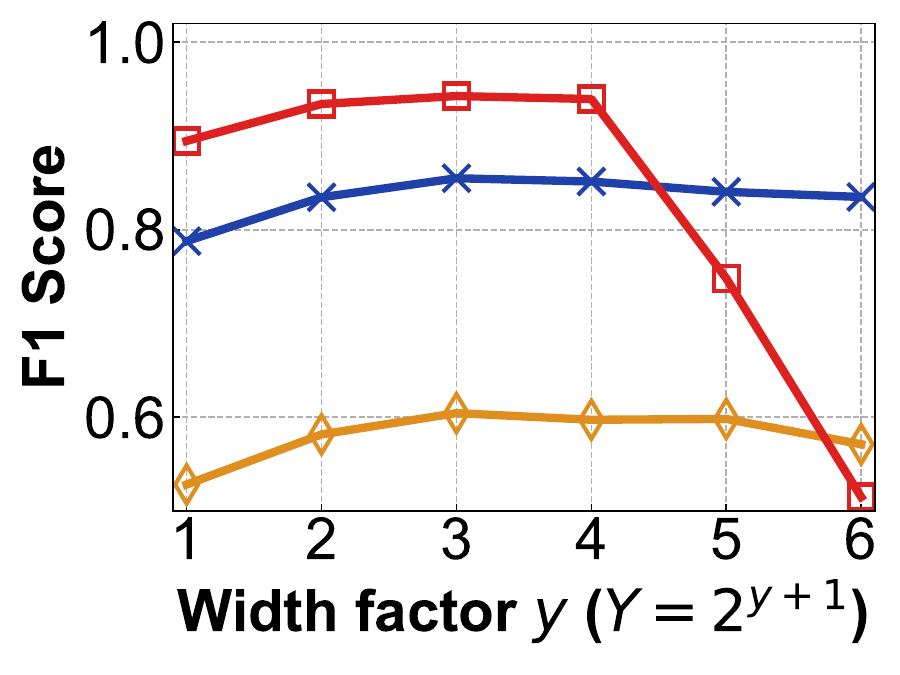}
        \caption{$Y$-F1,Campus}
        \label{xy_campus_f1}
        \vspace{0.25cm}
    \end{subfigure}
    \begin{subfigure}[b]{0.156\textwidth}
        \includegraphics[width=\textwidth]{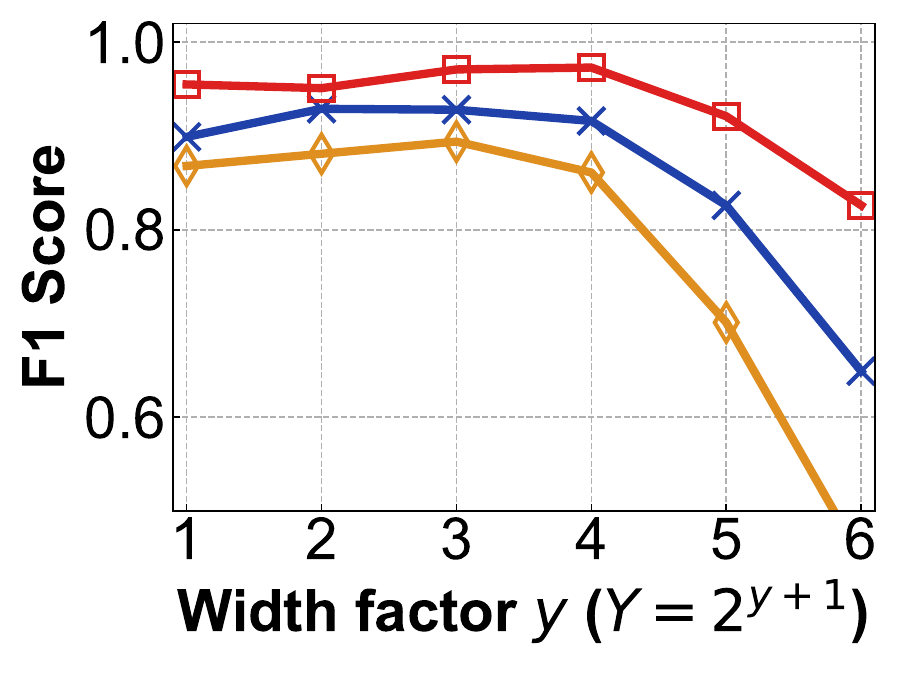}
        \caption{$Y$-F1,MAWI}
        \label{xy_mawi_f1}
        \vspace{0.25cm}
    \end{subfigure}

    \begin{subfigure}[b]{0.156\textwidth}
        \includegraphics[width=\textwidth]{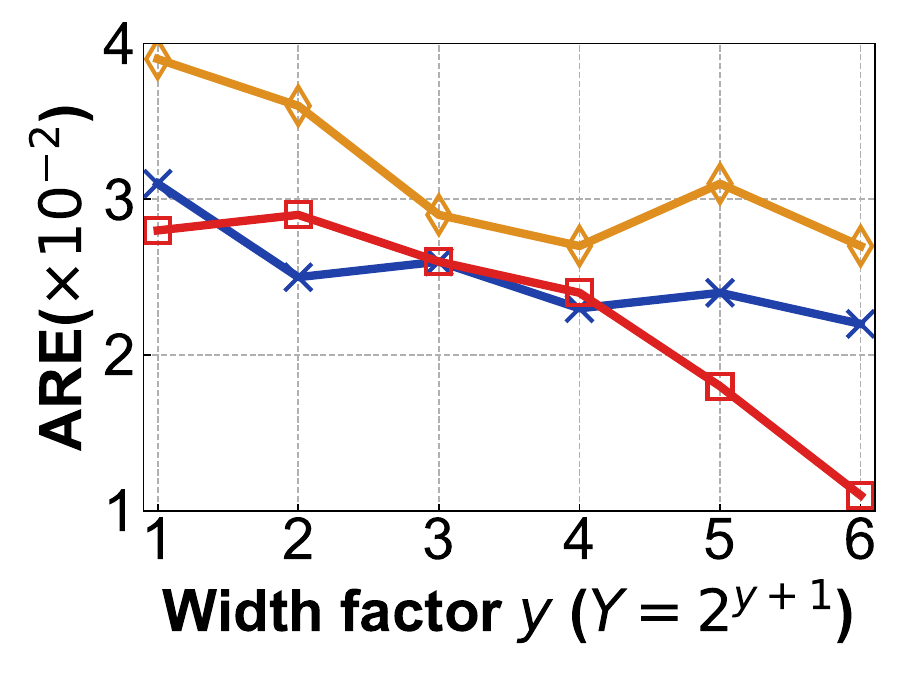}
        \caption{$Y$-ARE,CAIDA}
        \label{xy_caida_are}
        \vspace{0.25cm}
    \end{subfigure}
    \begin{subfigure}[b]{0.156\textwidth}
        \includegraphics[width=\textwidth]{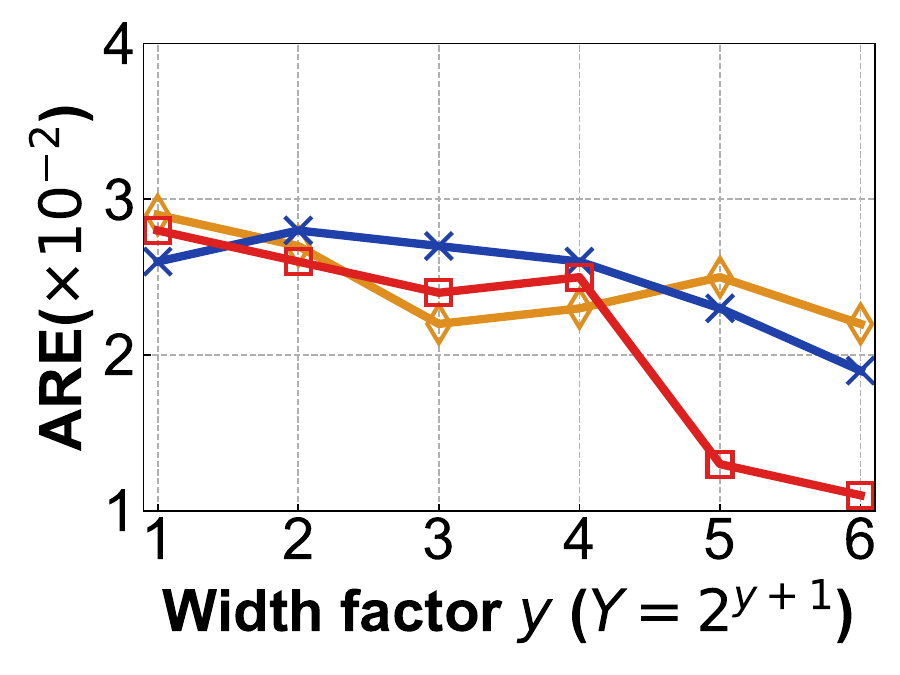}
        \caption{$Y$-ARE,Campus}
        \label{xy_campus_are}
        \vspace{0.25cm}
    \end{subfigure}
    \begin{subfigure}[b]{0.156\textwidth}
        \includegraphics[width=\textwidth]{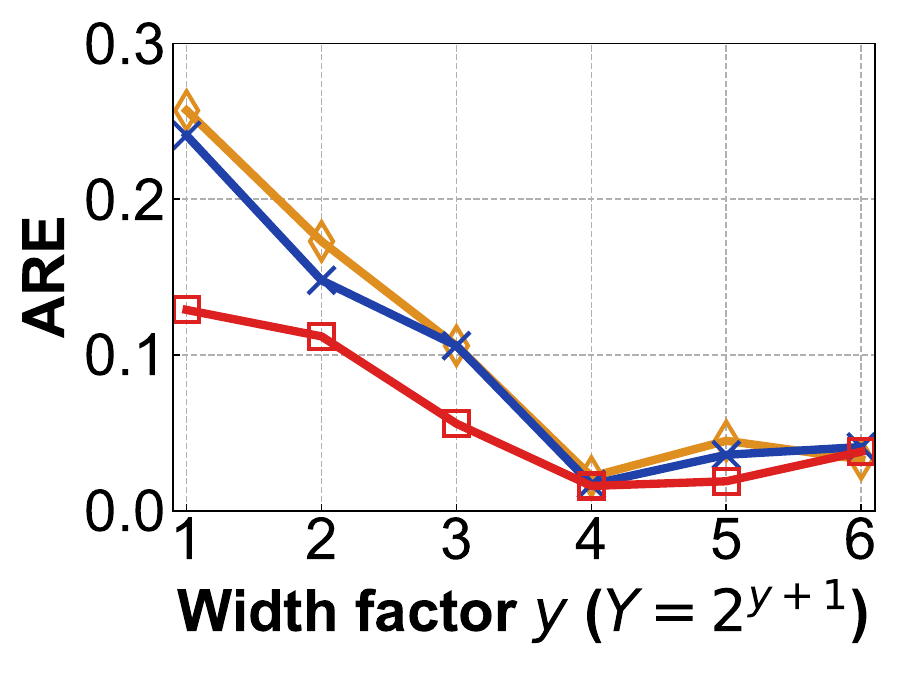}
        \caption{$Y$-ARE,MAWI}
        \label{xy_mawi_are}
        \vspace{0.25cm}
    \end{subfigure}

     \begin{subfigure}[b]{0.156\textwidth}
        \includegraphics[width=\textwidth]{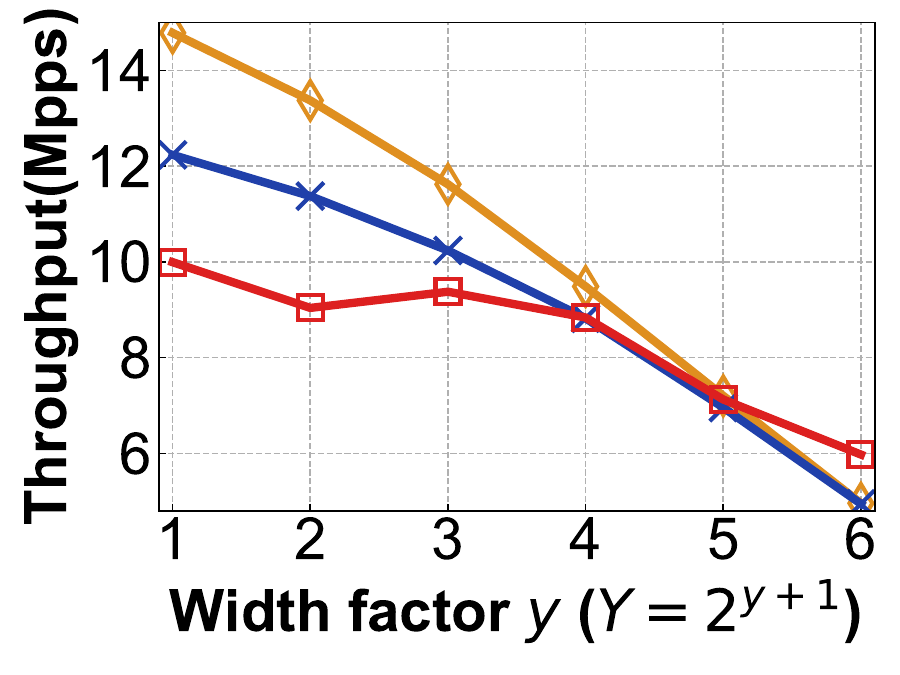}
        \caption{$Y$-Tp,CAIDA}
        \label{xy_caida_tp}
    \end{subfigure}
    \begin{subfigure}[b]{0.156\textwidth}
        \includegraphics[width=\textwidth]{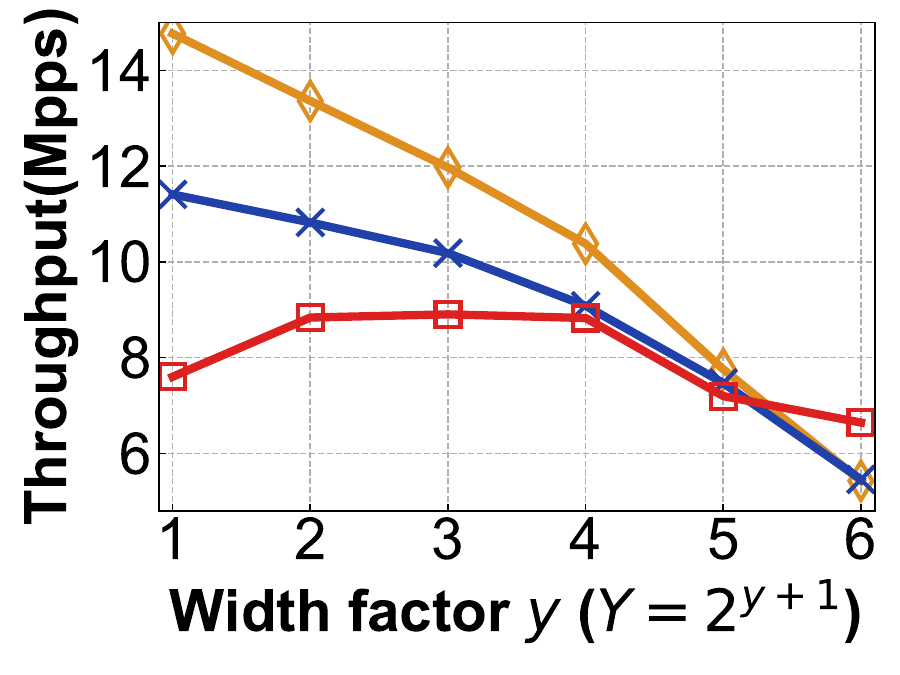}
        \caption{$Y$-Tp.,Campus}
        \label{xy_campus_tp}
    \end{subfigure}
    \begin{subfigure}[b]{0.156\textwidth}
        \includegraphics[width=\textwidth]{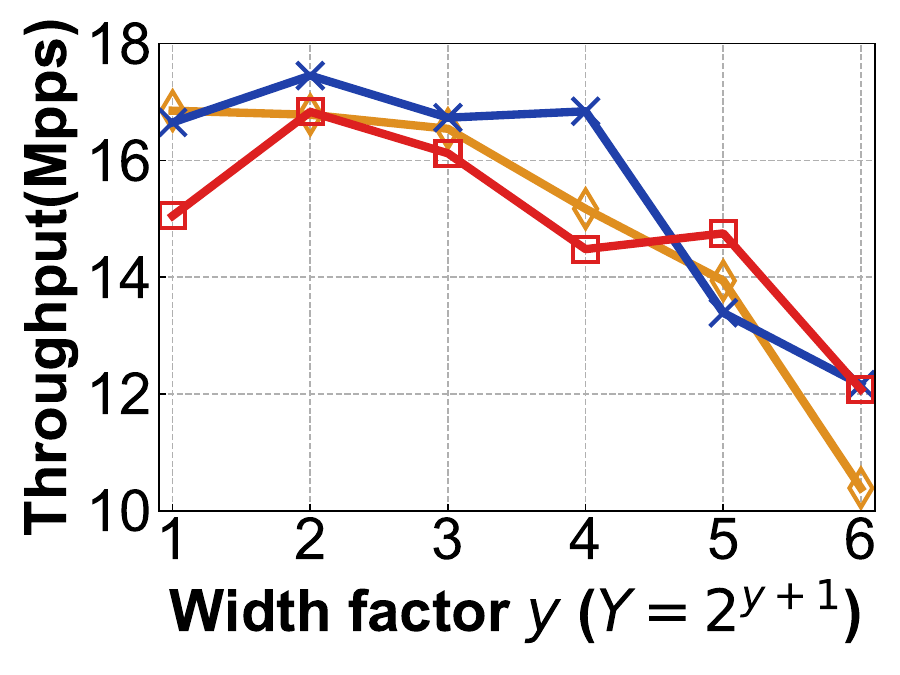}
        \caption{$Y$-Tp.,MAWI}
        \label{xy_mawi_tp}
    \end{subfigure}

    \caption{Impact of aspect ratio on F1, ARE and Throughput.}
    \label{exp_xy}
\end{figure}

\begin{figure*}[htbp]
    \centering
    \begin{subfigure}[b]{0.24\textwidth}
        \includegraphics[width=1\textwidth]{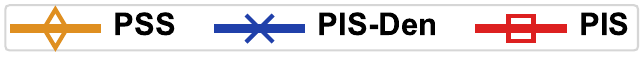}
        \caption*{}
        \vspace{-0.5cm}
        \label{}
    \end{subfigure}
    
    \begin{subfigure}[b]{0.163\textwidth}
        \includegraphics[width=\textwidth]{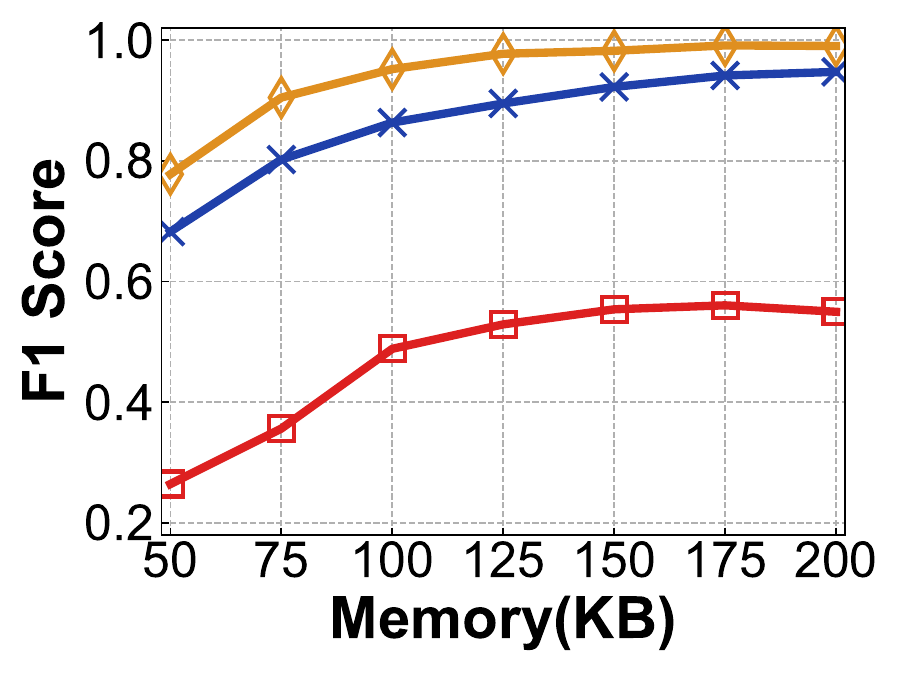}
        \caption{F1,CAIDA}
        \label{f1-caida}
        \vspace{0.3cm}
    \end{subfigure}
    \begin{subfigure}[b]{0.163\textwidth}
        \includegraphics[width=\textwidth]{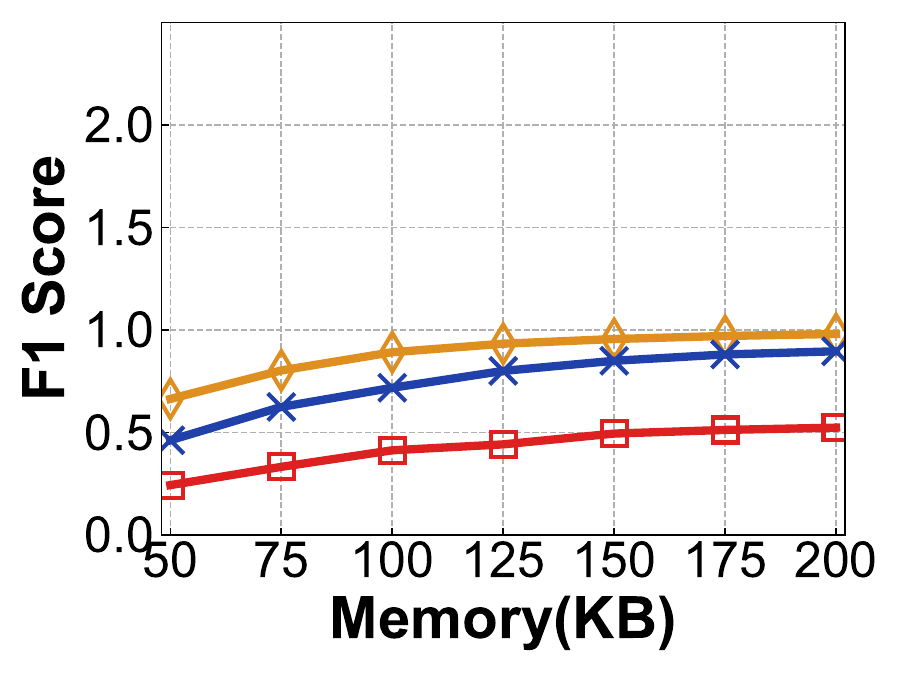}
        \caption{F1,Campus}
        \label{f1-campus}
        \vspace{0.3cm}
    \end{subfigure}
    \begin{subfigure}[b]{0.163\textwidth}
        \includegraphics[width=\textwidth]{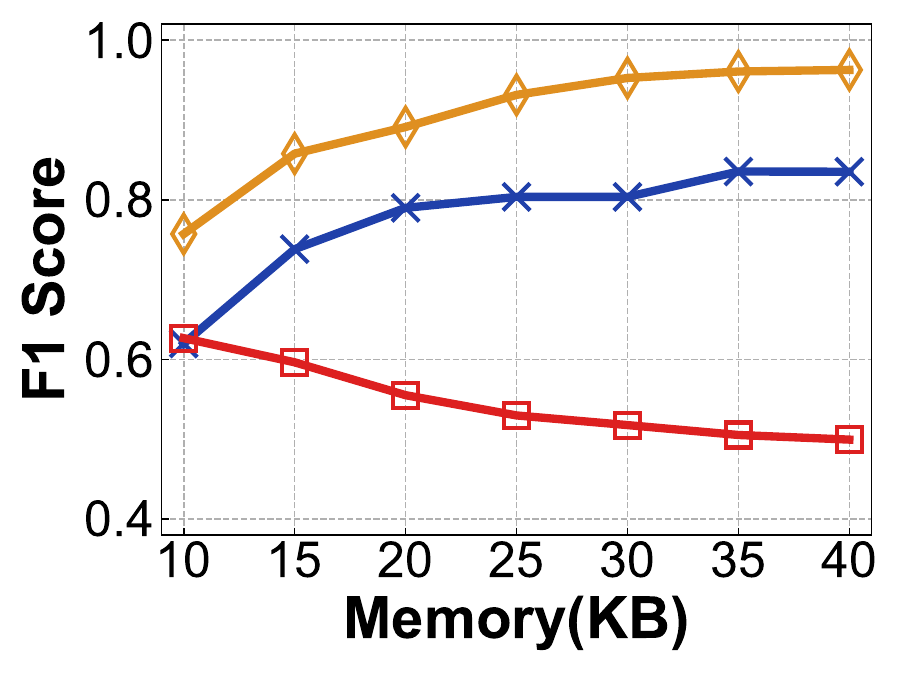}
        \caption{F1,MAWI}
        \label{f1-mawi}
        \vspace{0.3cm}
    \end{subfigure}   
    \begin{subfigure}[b]{0.163\textwidth}
        \includegraphics[width=\textwidth]{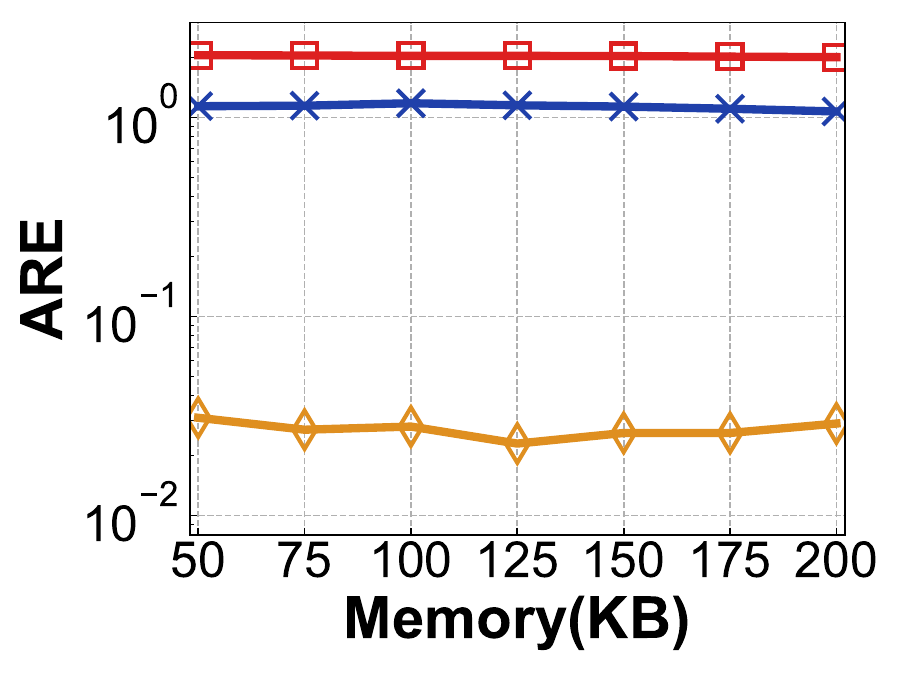}
        \caption{ARE,CAIDA}
        \label{are-caida}
        \vspace{0.3cm}
    \end{subfigure}
    \begin{subfigure}[b]{0.163\textwidth}
        \includegraphics[width=\textwidth]{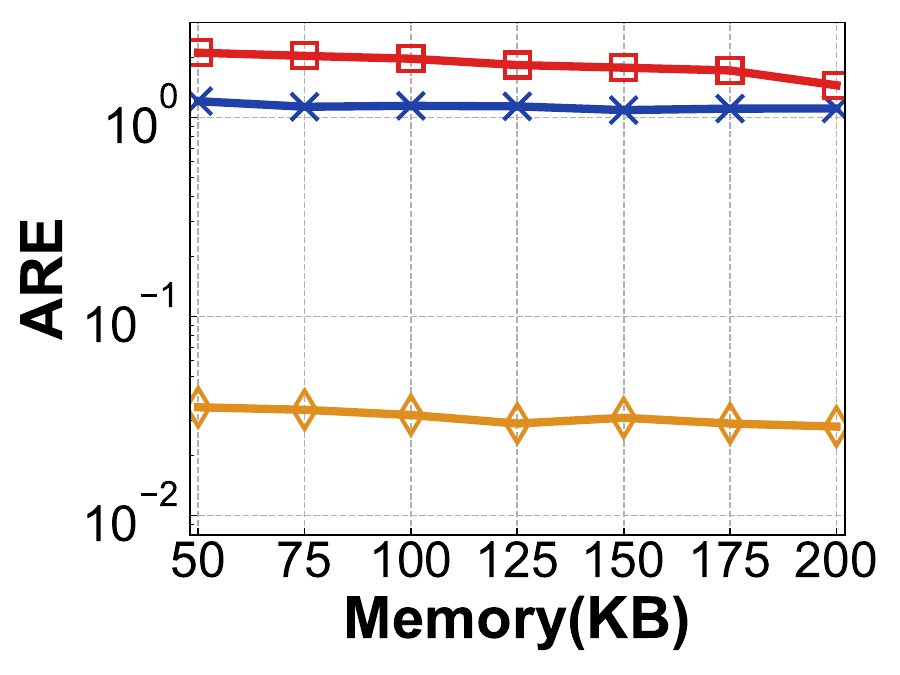}
        \caption{ARE,Campus}
        \label{are-campus}
        \vspace{0.3cm}
    \end{subfigure}
    \begin{subfigure}[b]{0.163\textwidth}
        \includegraphics[width=\textwidth]{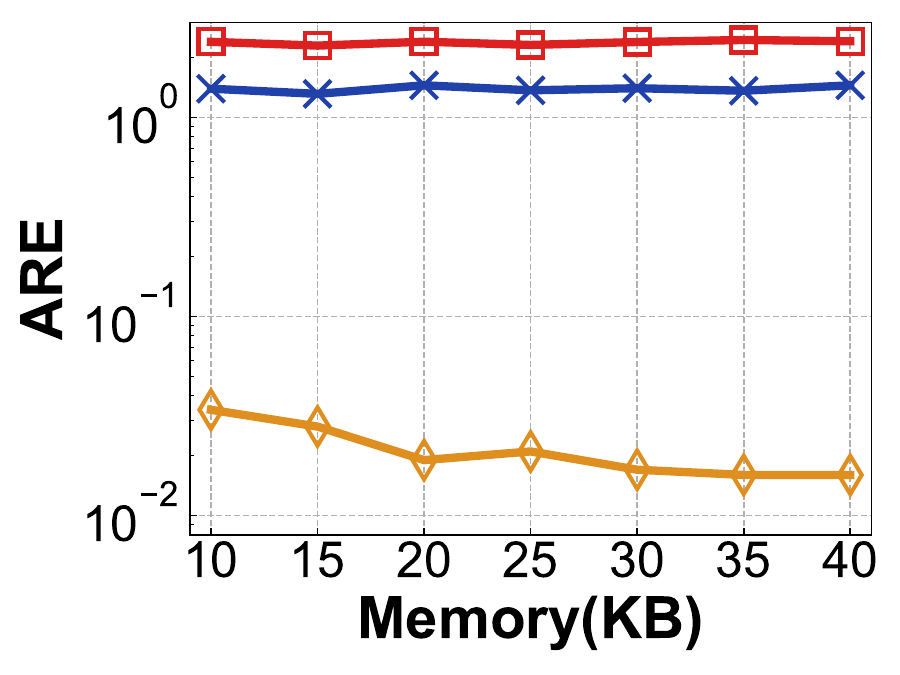}
        \caption{ARE,MAWI}
        \label{are-mawi}
        \vspace{0.3cm}
    \end{subfigure}   

\begin{subfigure}[b]{0.48\textwidth}
\centering
        \includegraphics[width=0.32\textwidth]{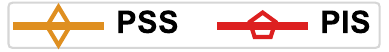}
        \caption*{}
        \vspace{-0.5cm}
        \label{}
    \end{subfigure}
    \begin{subfigure}[b]{0.48\textwidth}
\centering
        \includegraphics[width=0.57\textwidth]{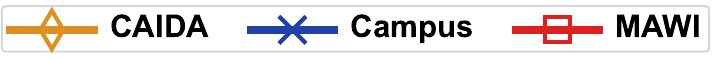}
        \caption*{}
        \vspace{-0.5cm}
        \label{}
    \end{subfigure}
   
\begin{subfigure}[b]{0.163\textwidth}
        \includegraphics[width=\textwidth]{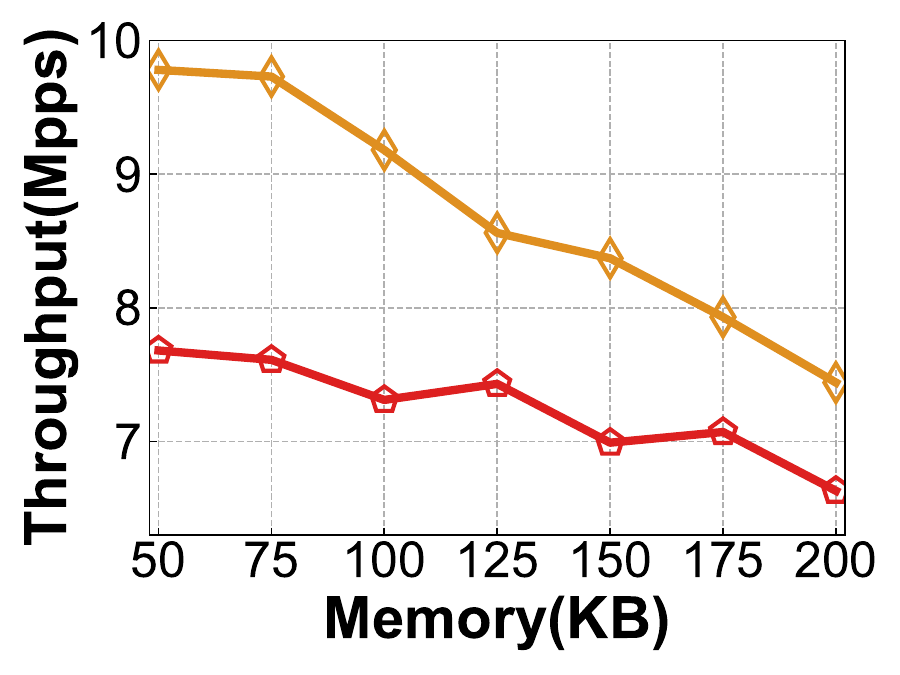}
        \caption{Tp,CAIDA}
        \label{tp_caida}
    \end{subfigure}
    \begin{subfigure}[b]{0.163\textwidth}
        \includegraphics[width=\textwidth]{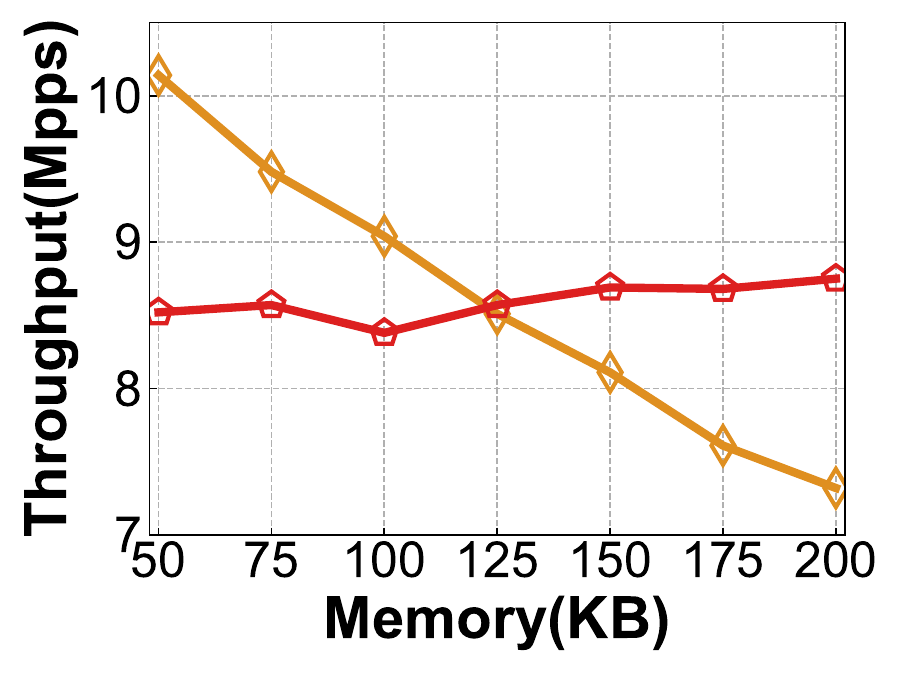}
        \caption{Tp,Campus}
        \label{tp_campus}
    \end{subfigure}
    \begin{subfigure}[b]{0.163\textwidth}
        \includegraphics[width=\textwidth]{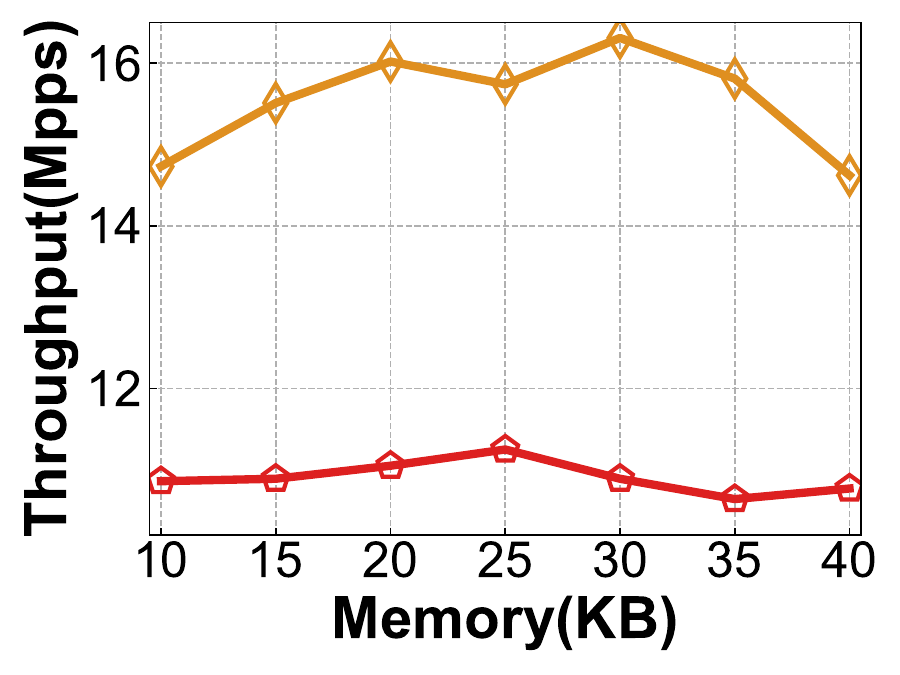}
        \caption{Tp,MAWI}
        \label{tp_mawi}
    \end{subfigure}   
    \begin{subfigure}[b]{0.163\textwidth}
        \includegraphics[width=\textwidth]{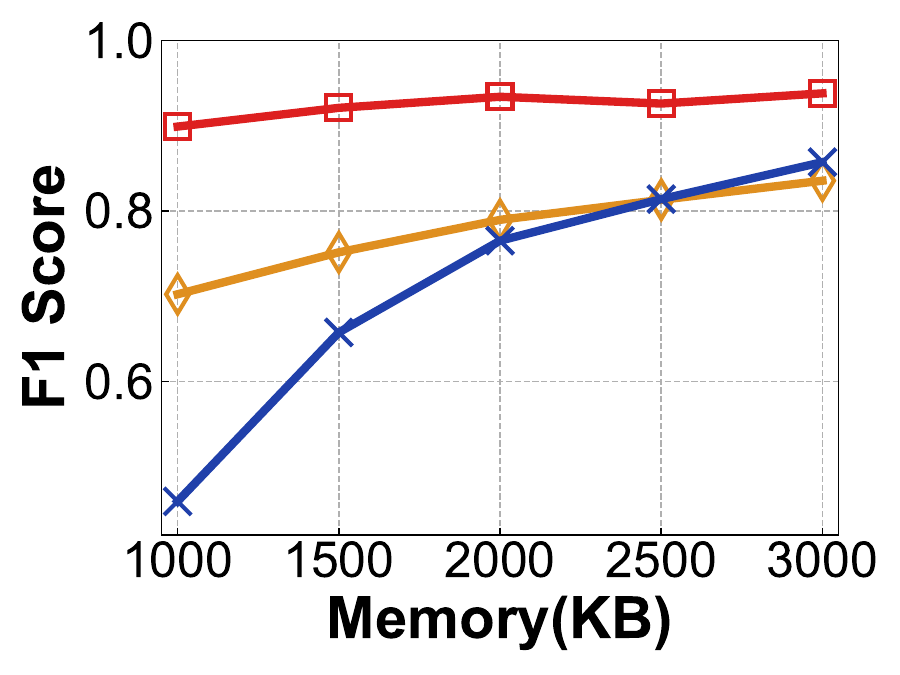}
        \caption{F1,Strawman}
        \label{Straw_f1}
    \end{subfigure}
    \begin{subfigure}[b]{0.163\textwidth}
        \includegraphics[width=\textwidth]{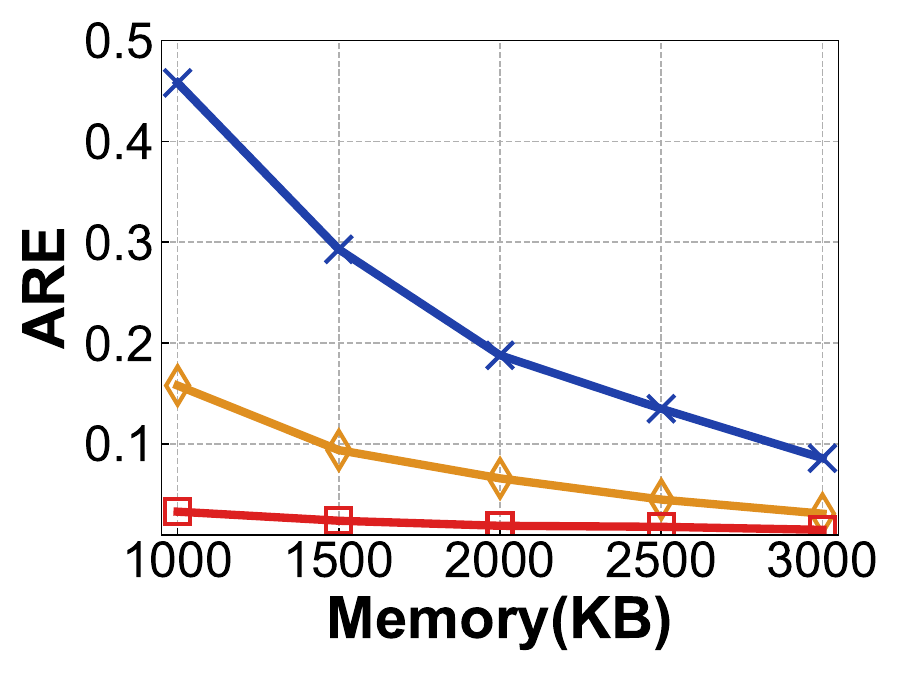}
        \caption{ARE,Strawman}
        \label{Straw_are}
    \end{subfigure}
    \begin{subfigure}[b]{0.163\textwidth}
        \includegraphics[width=\textwidth]{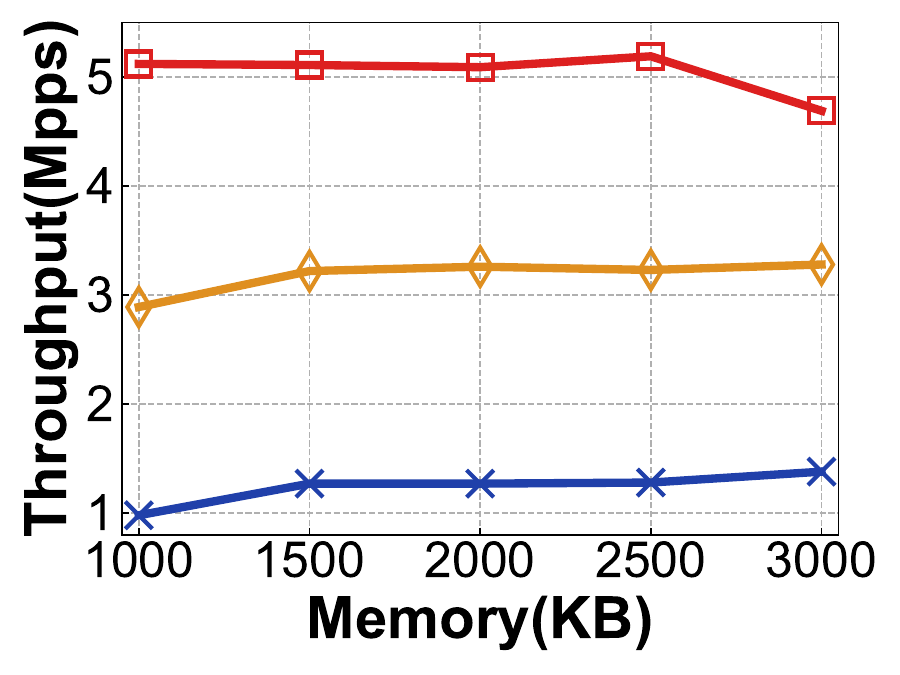}
        \caption{Tp,Strawman}
        \label{Straw_tp}
    \end{subfigure}   
    \caption{Comparison of performance.}
    \label{exp1}
    \vspace{-0.25cm}
\end{figure*}

    \subsection{Performance}

    In this experiment, we selected \( p_0 = 50 \), \( d_0 = 1.2 \), and \( Y = 32 \) as the testing parameters and compared the various metrics of PSSketch against other solutions. It is important to emphasize that the optimal parameter combination may vary depending on the dataset and application scenario. We cannot tailor the parameter selection specifically for the dataset we use. Thus, the parameters chosen in this experiment represent a relatively optimal level under the current dataset, while the potential performance ceiling is higher than the results we present.

    Figure \ref{exp1}(a), \ref{exp1}(b), \ref{exp1}(c) presents a comparison of the F1 Scores under varying memory constraints. In all three datasets, PSSketch significantly improved the accuracy of PS flow finding, with F1 Score raises to 1.80x-2.94x, 1.87x-2.73x, and 1.21x-1.93x, compared to PISketch, respectively. After modifying PISketch's definition of PS flow, the F1 Score still rises to 1.44x. This demonstrates that not only our model benefit from a more precise definition, but our data structure also exhibits clear advantages in filtering PS flows.

    Figure \ref{exp1}(d), \ref{exp1}(e), \ref{exp1}(f) displays the ARE under different memory conditions. Compared to PISketch, our method offers a 1-2 order of magnitude improvement. Specifically, under the higher load of the Campus dataset, the error decreases to 1.58\%-1.93\%, while in the lowest-load MAWI dataset, the error is reduced to 0.66\%-1.42\%. After adjusting the criterion for PISketch, we maintain at least a one-order-of-magnitude advantage in ARE, further supporting the high precision of PSSketch's statistical processing.

    Additionally, Figure \ref{exp1}(g), \ref{exp1}(h), \ref{exp1}(i) illustrates the throughput comparison between PSSketch and PISketch. Since PISketch-Density only differs in the filtering process of PS flows and does not alter the data processing, its throughput is identical to that of PISketch. Experimental results demonstrate that PSSketch achieves a significant advantage under medium to low loads across all memory constraints. Specifically, throughput in the CAIDA dataset reaches 1.28x, and in the MAWI dataset reaches 1.50x. Under higher loads, PSSketch outperforms PISketch in scenarios with strict memory limitations, achieving a throughput of 1.19x while being slightly slower than PISketch when memory is more abundant. Across the three datasets, we observe a decreasing trend in throughput as memory increases, particularly in higher-load conditions. Section 6.6 proposes an optimization to mitigate this issue, aiming for superior performance across all datasets, thereby surpassing PISketch.

    Lastly, Figure \ref{exp1}(j), \ref{exp1}(k), \ref{exp1}(l) presents the performance of the Strawman solution. Since its valid memory range differs by an order of magnitude from PSSketch and PISketch, its performance is reported separately. Within the 1MB-3MB memory range, Strawman achieves F1 Scores of 0.70-0.84, 0.46-0.86, and 0.90-0.94 with the three datasets, notably lower than PSSketch. The corresponding ARE values are 1.11x-6.87x, 2.46x-16.36x, and 0.54x-2.06x compared to PSSketch. Although Strawman attains a higher F1 Score than PISketch and its ARE is generally on the same order of magnitude as that of PSSketch, it requires more than ten times the memory, and its throughput is only 13.39\%–31.39\% of that of PSSketch and 15.77\%–45.19\% of that of PISketch.

    \subsection{SIMD Optimization}

    To curb the trend of throughput decreasing with increasing memory, we utilize the SIMD technology in the process of traversing the Competition Layer, which involves two main operations:

    \textit{Window Flag Reset:} At the beginning of each window, we need to reset the flag \( W \) of all entries in the \( M \) buckets to zero. By using SIMD, the reset for each bucket can be completed within 1-2 memory accesses, significantly reducing the memory access time.

    \textit{Bucket Search:} For each flow \( e \), the insertion process requires searching the corresponding bucket \( BK[m] \). As discussed in Section 4.3, we previously optimized the three-loop process into a single loop. With the implementation of SIMD, this loop can be further reduced to 1-2 memory accesses, enhancing the search efficiency.

    As shown in Figure \ref{exp2}, the throughput of PSSketch on three different datasets improves up to 1.68x, 1.48x, 1.85x respectively after integrating SIMD. More importantly, we eliminate the trend of decreasing throughput as memory size increases. That means, PSSketch demonstrates superior speed to the original PSSketch and strawman across all memory constraints and load intensities.

\begin{figure}[htbp]
    \centering
    \begin{subfigure}[b]{0.2\textwidth}
        \includegraphics[width=1\textwidth]{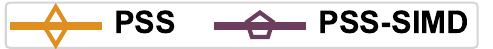}
        \caption*{}
        \vspace{-0.5cm}
        \label{}
    \end{subfigure}
    
    \begin{subfigure}[b]{0.156\textwidth}
        \includegraphics[width=\textwidth]{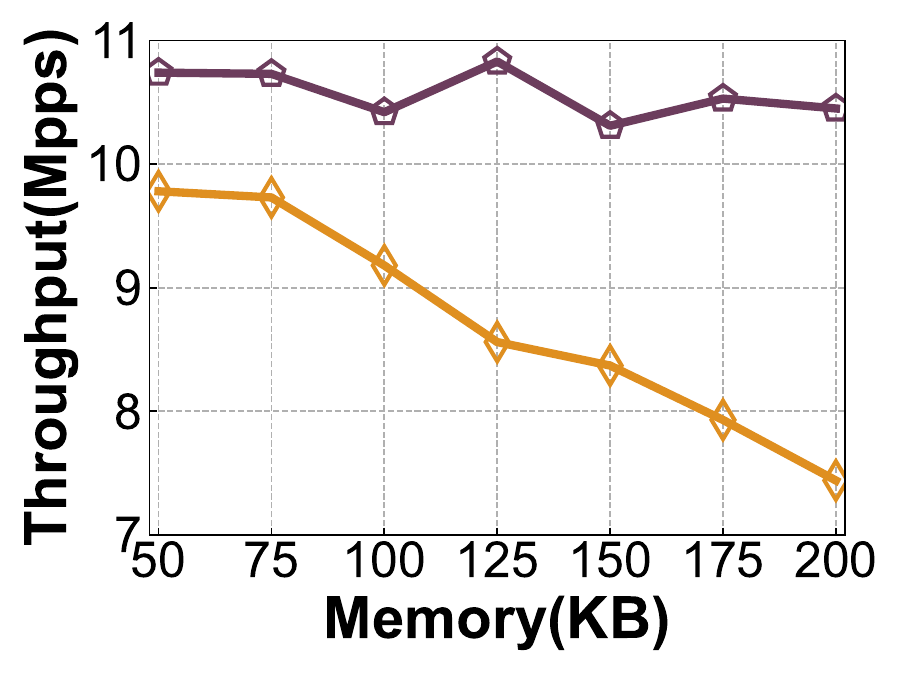}
        \caption{Tp,CAIDA}
        \label{simd_caida}
    \end{subfigure}
    \begin{subfigure}[b]{0.156\textwidth}
        \includegraphics[width=\textwidth]{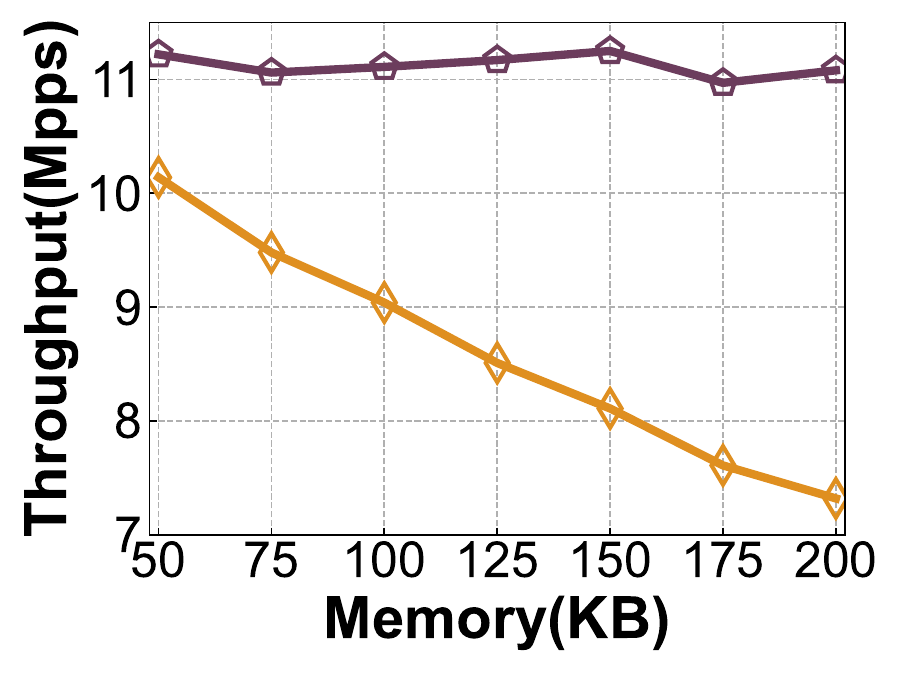}
        \caption{Tp,Campus}
        \label{simd_campus}
    \end{subfigure}
    \begin{subfigure}[b]{0.156\textwidth}
        \includegraphics[width=\textwidth]{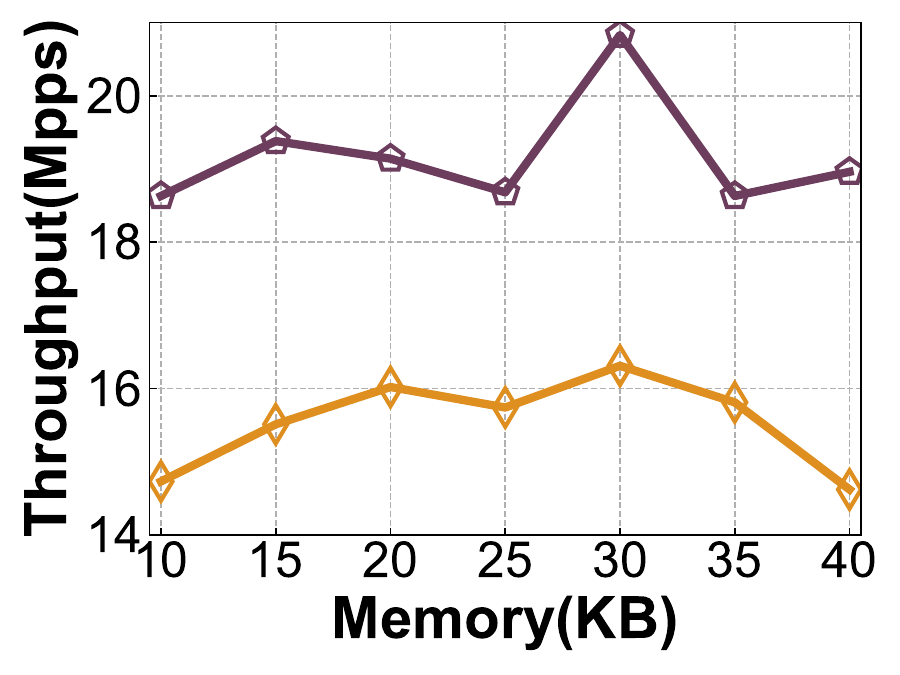}
        \caption{Tp,MAWI}
        \label{simd_mawi}
    \end{subfigure}
    \caption{The performance of PSSketch using SIMD.}
    \vspace{-0.35cm}
    \label{exp2}
\end{figure}

\section{Conclusion}
    Finding persistent sparse flow is essential for network threat detection. In this paper, we first observe the distribution characteristics of different datasets. Then, we introduce a more accurate criterion for PS flows and propose PSSketch, a highly precise data structure. PSSketch employs a dual-layer structure with variable-length bit-level counters; it records frequency and persistence in the Competition Layer and counts overflow times in the Protection Layer, significantly enhancing memory efficiency. We use isolation protection and probability replacement to effectively protect potential PS flows, making it hard to replace them with a large number of regular flows. Additionally, optimizations such as pruning, burst elimination, and one-time traversal are applied to ensure high throughput. Experiments demonstrate that, compared to combined solutions and SOTA method PISketch, we achieve up to 2.94x in F1 Score and a reduction in ARE by 1-2 orders of magnitude. Furthermore, our method has higher throughput across all scenarios. Overall, PSSketch provides a novel solution for finding persistent sparse flows with both high precision and efficiency.

\bibliographystyle{ACM-Reference-Format}
\bibliography{ref.bib}

\clearpage

\appendix

\section{Detail version of definitions}

\textbf{Flow:} Suppose $E$ is an item set containing all possible combinations of meta-information for a network packet. To simplify the problem, we assume that each element contains only the ID of a flow, and denote it as $e$. A data flow $S$ is a multiset of $N$ elements $\langle e_1,e_2,...,e_N \rangle (e_i \in E)$.  

    \textbf{Frequency:} In a data flow $S$, the frequency of an element $e_i \in S$ refers to its multiplicity, denoted as $f_i$. That means $e_i$ appears $f_i$ times in the $S$.

    \textbf{Time Window:} Given a data flow $S=\{ e_1,e_2,e_3,…,e_N \} $ and a window size $t$, the elements of $S$ can be partitioned into multiple subsets, where $T_0=\{ e_1,e_2,…,e_t\}$, $T_1=\{e_{t+1},e_{t+2},…,e_{2t} \} $, and so forth. Each subset $T_k$ contains exactly $t$ elements, except for the last subset, which contains no more than $t$ elements.(The partitioning is defined as a collection of $multisets$, where elements can appear multiple times in a particular window in the data flow.) An element $e_i$ belongs to the $k^{th}$ window if $e_i \in T_k$, or we say the element $e_i$ appears in the $k^{th}$ window of the data flow. Mathematically, the partitioning is defined as:

\begin{equation}
T_k = \{e_{kt+1}, e_{kt+2}, \dots, e_{(k+1)t}\} \quad \text{for} \quad 0 \leq k < \left\lceil \frac{n}{t} \right\rceil.
\end{equation}
    
    \textbf{Persistence:} Let $e_i$ be an element in a data flow $T=\langle T_0,T_1,...\rangle$ which is partitioned by time windows. There exists a subset $T_i=\langle T_{i0},T_{i1},...\rangle$ such that $e_i \in T_{i0}$, $e_i \in T_{i1},...$ while $e_i$ does not appear in any time window in set $T\backslash T_i$. We denote $|T_i|$, the number of elements in $T_i$, as the persistence of $e_i$. That indicates the number of distinct time windows where $e_i$ appears at least once.

    \textbf{Density:} Given a continuous sequence of \( \tau \) windows, if the element \( e_i \) exists in at least one window among \( T_m, T_{m+1}, \dots, T_{m+\tau-1} \), its density is defined as:
    
\begin{equation}
D^\tau_i = \frac{f^\tau_i}{p^\tau_i},
\end{equation}

where \( p^\tau_i \) denotes the persistence of element \( e_i \) over the \( \tau \) windows, representing the number of windows in which \( e_i \) appears, with \( p^\tau_i > 0 \), and \( f^\tau_i \) represents the frequency of element \( e_i \) in these \( \tau \) windows, where \( f^\tau_i \geq p^\tau_i \). It reflects the average access intensity of the element \( e_i \) within the data flow over the specified \( \tau \) windows.

\input{appendix_math}

\newpage
\section{Extra Figure and Pseudocode}
\label{app_extra}

\begin{figure}[H]
        \centering
            \includegraphics[width=\linewidth]{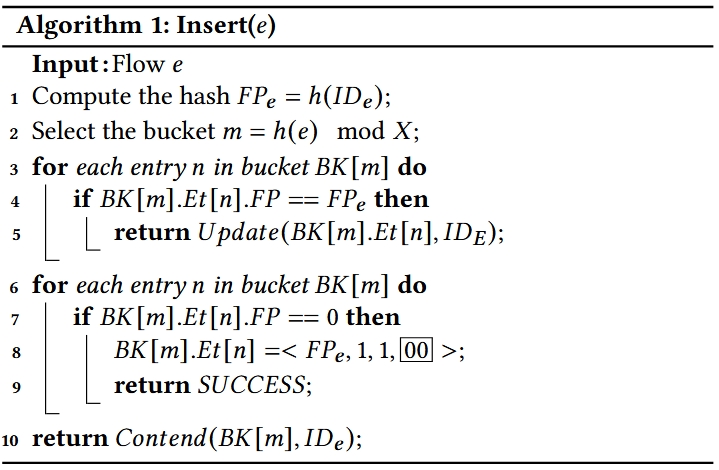}   
        \caption{Insert Algorithm.}
        \label{str}
    \end{figure}

    \begin{figure}[H]
        \centering
            \includegraphics[width=\linewidth]{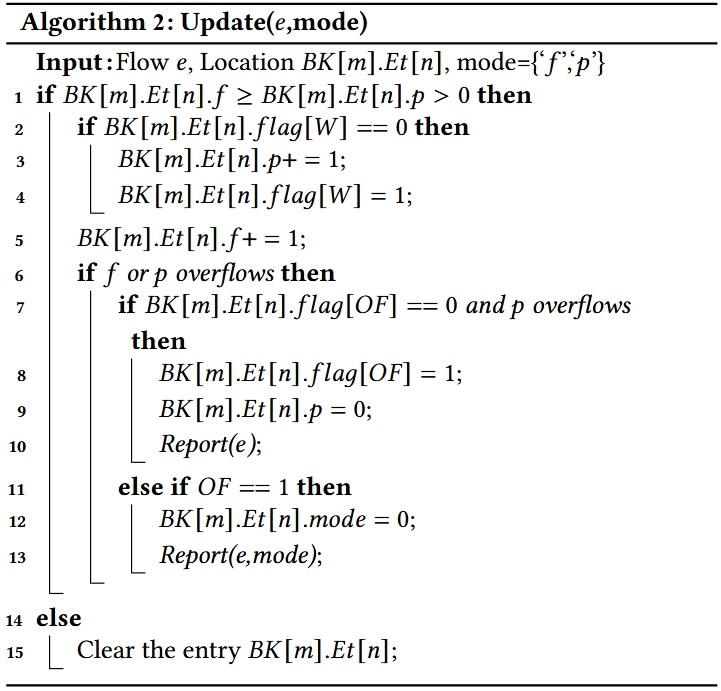}   
        \caption{Update Algorithm.}
        \label{str}
    \end{figure}





\begin{figure}[H]
    \centering
    \begin{subfigure}[b]{0.2\textwidth}
        \includegraphics[width=\textwidth]{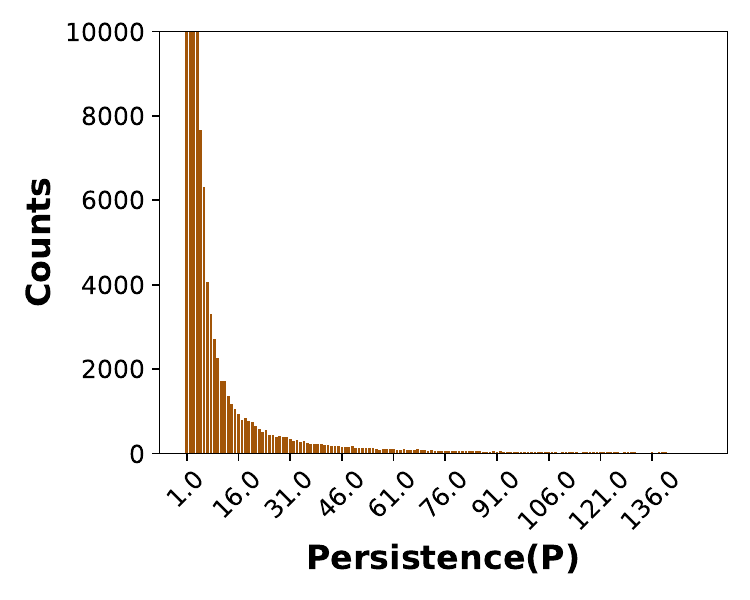}
        \caption{Persistence}
        \label{888}
    \end{subfigure}
    \begin{subfigure}[b]{0.2\textwidth}
        \includegraphics[width=\textwidth]{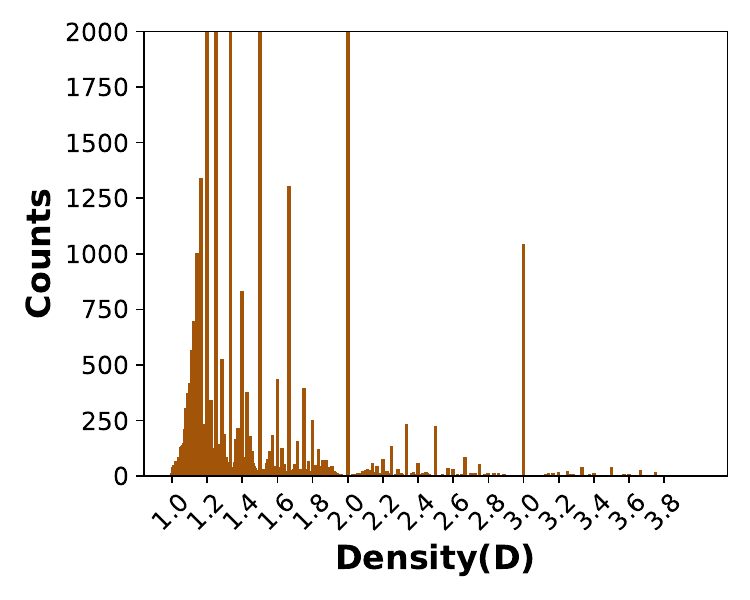}
        \caption{Density}
        \label{777}
    \end{subfigure}
    \caption{The distribution characteristics of Campus.}
    \label{444}
\end{figure}

\begin{figure}[H]
    \centering
    \begin{subfigure}[b]{0.2\textwidth}
        \includegraphics[width=\textwidth]{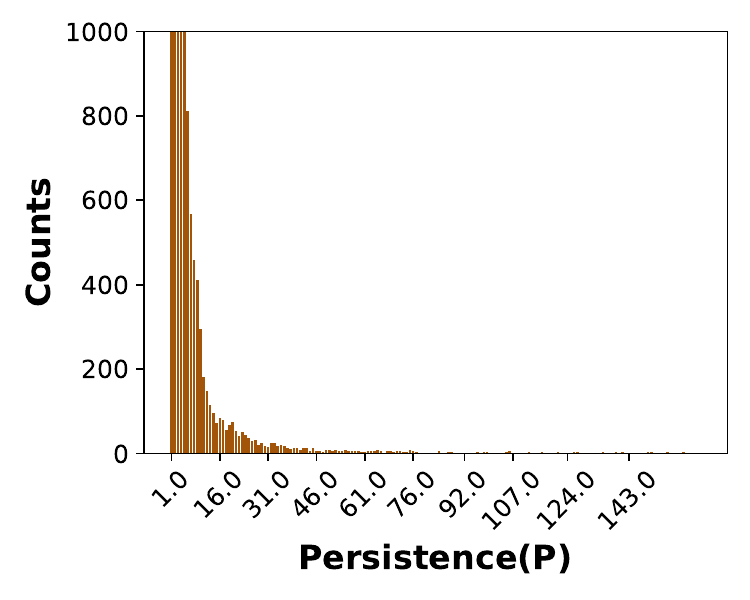}
        \caption{Persistence}
        \label{555}
    \end{subfigure}
    \begin{subfigure}[b]{0.2\textwidth}
        \includegraphics[width=\textwidth]{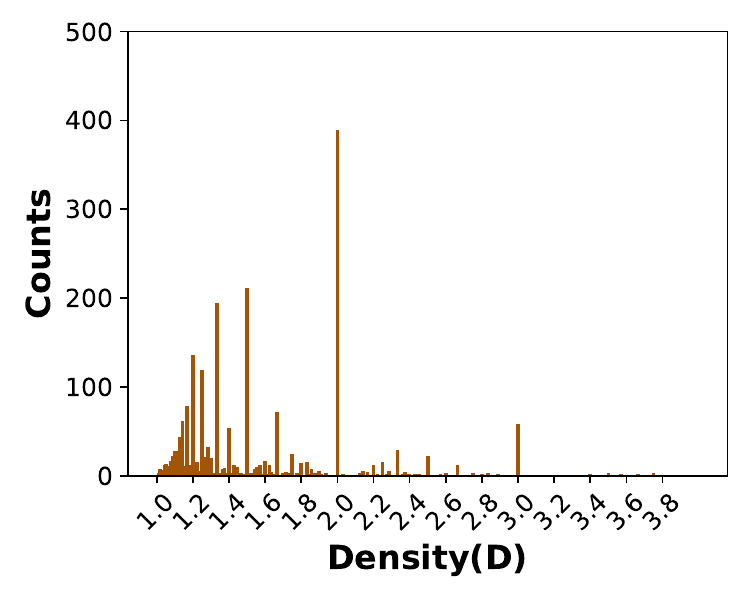}
        \caption{Density}
        \label{666}
    \end{subfigure}
    \caption{The distribution characteristics of MAWI.}
    \label{999}
\end{figure}

\begin{figure}[H]
    \centering
    \begin{subfigure}[b]{0.2\textwidth}
        \includegraphics[width=\textwidth]{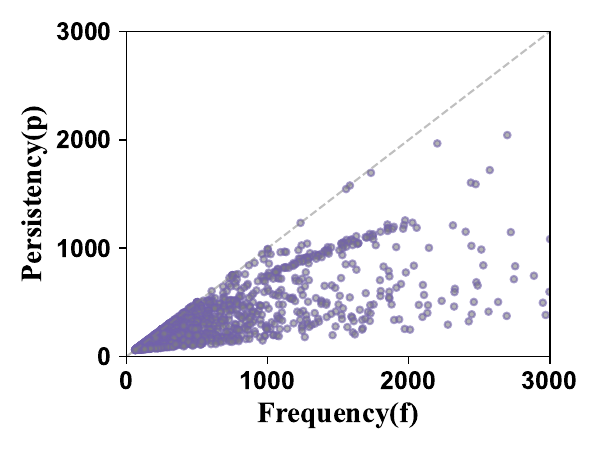}
        \caption{PISketch}
        \label{qwe}
    \end{subfigure}
   \begin{subfigure}[b]{0.2\textwidth}
        \includegraphics[width=\textwidth]{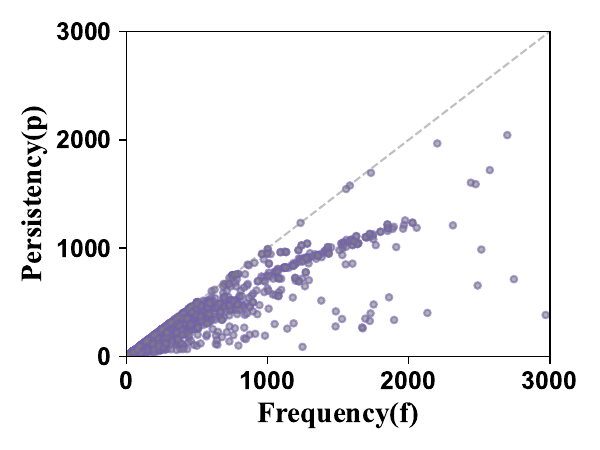}
        \caption{PSSketch}
        \label{wer}
    \end{subfigure}
    \caption{Top-2K PS flows reported under the CAIDA dataset.}
    \label{ert}
\end{figure}

\begin{figure}[H]
    \centering
    \begin{subfigure}[b]{0.2\textwidth}
        \includegraphics[width=\textwidth]{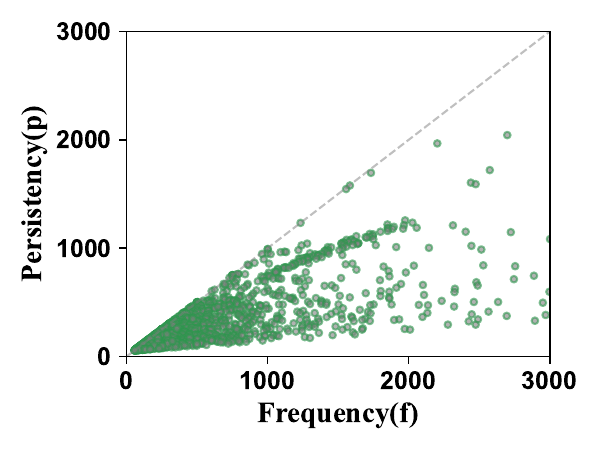}
        \caption{PISketch}
        \label{rty}
    \end{subfigure}
   \begin{subfigure}[b]{0.2\textwidth}
        \includegraphics[width=\textwidth]{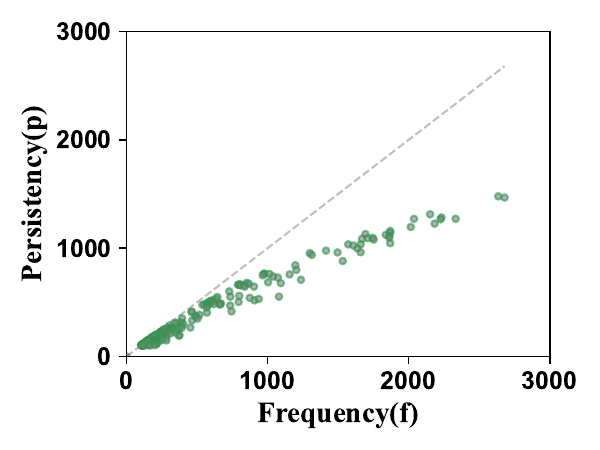}
        \caption{PSSketch}
        \label{tyu}
    \end{subfigure}
    \caption{Top-2K PS flows reported under the MAWI dataset.}
    \label{yui}
\end{figure}

\end{document}

%% file: Mathematical_Analysis.tex
\label{mathematical analysis}

In this section, we first analyze the property of the density of a flow in \ref{Property of the density of a Flow}. Then, we derive the error bounds for the PSSketch in \ref{Error Bound Analysis}. Finally, we analyze the time complexity and space complexity of PSSketch in \ref{Complexity analysis}.All the proof details can be found in \ref{Details of mathematical analysis}.

\begin{figure}[htbp]
    \centering
    \includegraphics[width=220pt]{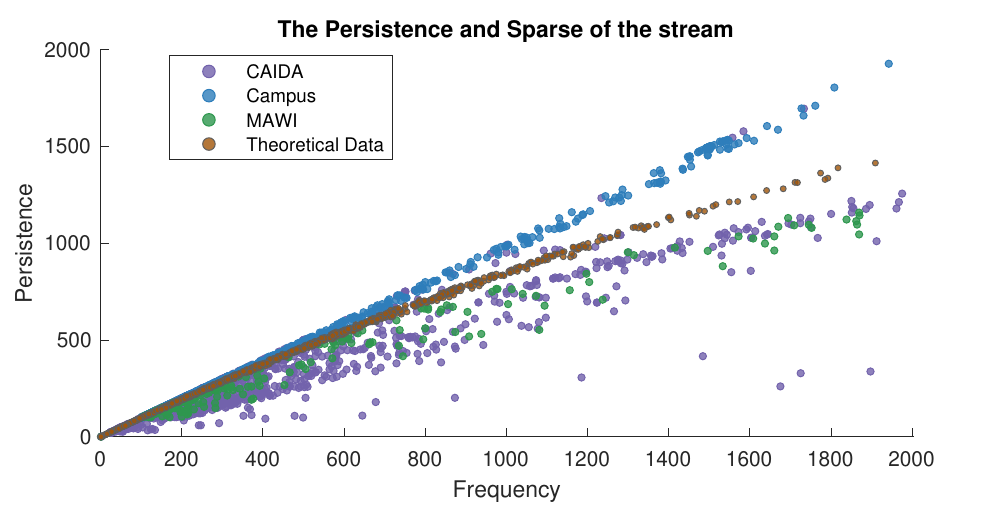}
    \caption{The persistence and frequency of Theoretical data and the data of CAIDA, Campus and MAWI.}
    \label{fpc}
    \vspace{-0.35cm}
\end{figure}

In this section, we analyze the implementation of PSSketch and make the following two assumptions:
\begin{enumerate}
    \item The majority of flows in the data stream are neither persistent nor sparse (see \ref{PS_define} for supporting evidence). 
    Consequently, PS flows constitute only a small fraction of the data stream.
    \item We assume that all flows are mutually independent. For each flow, it is independent and identically distributed (i.i.d.) within each window, following a $Poisson$ distribution with parameter \(\lambda\), where the parameter \(\lambda\) follows a normal distribution. Figure \ref{fpc} demonstrates that our assumptions align with the flow distribution in the three datasets we used (i.e., CAIDA, Campus, and MAWI).
\end{enumerate}

The symbols used in this section are listed in Table \ref{symbols}.

\subsection{Property of the Density of a Flow}
\label{Property of the density of a Flow}

Based on the given assumption, for a specific flow \( e \) (with the corresponding $Poisson$ distribution parameter \( \lambda \)), let \( n_i \) represent the times that flow $e$ appears in the \( i^{th} \) window.

\begin{theorem}
    The expectation and variance of frequency, persistence, and density of \(e\) after the \( i^{th} \) window is given by:
    \begin{equation}
    \begin{split}
        \mathbb{E}[f_i]&=i\lambda \\
        \mathbb{E}[p_i]&=i\cdot\left(1-e^{-\lambda}\right)\\
        \mathbb{E}[d_i]&=\frac{\lambda}{\left(1-e^{-\lambda}\right)}.
    \end{split}
    \end{equation}
        \begin{equation}
    \begin{split}
        \mathbb{VAR}[f_i]&=i\lambda\\
        \mathbb{VAR}[p_i]&=i\cdot e^{-\lambda}\cdot\left(1-e^{-\lambda}\right)\\
        \mathbb{VAR}[d_i]&\leq \frac{\lambda^2}{\left(1-e^{-\lambda}\right)^2}+\frac{\lambda+\lambda^2}{\left(1-e^{-\lambda}\right)}<\infty.
    \end{split}
    \end{equation}
\end{theorem}

According to \(L_p\) Convergence Theorem and 
Dominated Convergence Theorem, we can acknowledge that 
\begin{equation}
    \lim_{i\rightarrow \infty}\mathbb{E}[|d_i-\lambda|^2]=0,
\end{equation}
which demonstrates that \(d_i\) is \(L_2\) and almost surely converge to \(\lambda\).

Due to the constraint on \( p_0 \), the \( \lambda \) parameter of the Poisson distribution in the PS flows should neither be too small nor excessively large, balancing the need to identify low-density flows.

\subsection{Error Bound Analysis}
\label{Error Bound Analysis}

For PSSketch, the sources of error come from \textbf{Ejection Error}.




\begin{theorem}
 \( \hat{d_i} \) is an unbiased estimator of \( \lambda \), and \( \hat{d_i} \) \( L_2 \)-convergence as well as almost surely convergence to \( \lambda \). 

\end{theorem}

\begin{theorem}
 \( \hat{d_i}-d_i \) is an unbiased estimator of \( 0 \), and \( \hat{d_i}-d_i \)  \( L_2 \)-convergence as well as almost surely convergence to \( 0 \). 
\end{theorem}

This demonstrates that, when the number of windows is sufficiently large, the error caused by ejection converges to zero. This means that even if a suspected PS flow is ejected, it will eventually be reported as a PS flow as the number of windows is large enough.

\subsection{Complexity Analysis}
\label{Complexity analysis}

\subsubsection{Time Cost of PSSketch}

In the Insert phase, a hash function is applied once for both the competition and protection layers, each with a time complexity of \( O(1) \).

\subsubsection{Space Cost of PSSketch}

\begin{theorem}
    The total storage space required by PSSketch is 
\begin{equation}
    XY \cdot (L_c^{FP} + L_c^{f} + L_c^{p} + L_c^{flag}) + R \cdot (L_c^{ID} + L_c^{f_{of}} + L_c^{p_{of}}),
\end{equation}
where \( R \ll X \cdot Y \).
\end{theorem}

\begin{theorem}
    The maximum frequency and the maximum persistence of a flow that can be stored in PSSketch is 
    \begin{equation}
        f_{max} = (2^{L_c^{f_{of}}} - 1) \cdot (2^{L_c^{f}} - 1) \approx 2^{L_c^{f_{of}} + L_c^{f}},
    \end{equation}
\begin{equation}
    p_{max} = (2^{L_c^{p_{of}}} - 1) \cdot (2^{L_c^{p}} - 1) \approx 2^{L_c^{p_{of}} + L_c^{p}}.
\end{equation}
\end{theorem}

\begin{theorem}
    If PISketch stores flow data with the same order of magnitude for the maximum frequency and persistence as PSSketch, then the Weight Sketch component of PISketch must have at least:
    \begin{equation}
        XY \cdot \left( L_c^{ID} + \log_2 L \cdot \left( L_c^{f_{of}} + L_c^{f} \right) + \left( L_c^{p_{of}} + L_c^{p} \right) \right)
    \end{equation}
\end{theorem}

\begin{theorem}
    If the Weight Sketch in PISketch and the Competition Layer in PSSketch occupy the same storage space, then:
    \begin{equation}
        p_{max}^{PISketch} = 2^{L_c^{N}} - 1 \ll p_{max}.
    \end{equation}
\end{theorem}

The above arguments demonstrate that, compared to PISketch, PSSketch has a significant advantage in space complexity due to its handling of overflows and the use of a dual-layer mechanism consisting of the Competition Layer and Protection Layer.




%% file: appendix_math.tex
\section{Details of mathematical analysis}
\label{Details of mathematical analysis}

The symbols used in this section are listed in Table \ref{symbols}.

\begin{table}[]
    \centering
    \begin{tabular}{cc}
        \toprule
        \textbf{Symbol} & \textbf{Meaning} \\
        \midrule
        \( e \) & An arbitrary flow in \( S \) \\
        \( \lambda \) & Poisson parameter of \( e \) \\
        \( n_i \) & The number of times the flow \( e \) appears in the \( i^{th} \) window \\
        \( f_i \) & The frequency of \( e \) in the first \( i \) windows \\
        \( p_i \) & The persistence of \( e \) in the first \( i \) windows \\
        \( d_i \) & The density of \( e \) in the first \( i \) windows \\
        \( \hat{d_i} \) & The estimated density of \( e \) in the first \( i \) windows \\
        \( K_t \) & The last time flow \( e \) was kicked, occurring in the \( t^{th} \) window \\
        \( L_c^{FP} \) & Bit-width for storing flow fingerprints \\
        \( L_c^f \) & Bit-width for storing flow frequency counters \\
        \( L_c^p \) & Bit-width for storing flow persistence counters \\
        \( L_c^{flag} \) & Bit-width for storing flow status flags \\
        \( L_c^{ID} \) & Bit-width for storing flow identifiers \\
        \( L_c^{f_{of}} \) & Bit-width for storing frequency counter overflows \\
        \( L_c^{p_{of}} \) & Bit-width for storing persistence counter overflows \\
        \( L_c^W \) & Bit-width for weight counters \\
        \( R \) & Number of cells in the protection layer \\
        \( X \) & Number of buckets in the competition layer \\
        \( Y \) & Number of cells per bucket in the competition layer \\
        \( f_{max} \) & Maximum frequency of a flow storable in PSSketch \\
        \( p_{max} \) & Maximum persistence of a flow storable in PSSketch \\
        \( f_{max}^{PISketch} \) & Maximum frequency of a flow storable in PISketch \\
        \( p_{max}^{PISketch} \) & Maximum persistence of a flow storable in PISketch \\
        \bottomrule
    \end{tabular}
    \caption{Symbols used in \ref{mathematical analysis}}
    \label{symbols}
\end{table}

\subsection{Property of the density of a Flow}

Based on the given assumption, for a specific flow \( e \) (with the corresponding $Poisson$ distribution parameter \( \lambda \)), let \( n_i \) represent the times that flow $e$ appears in the \( i^{th} \) window.

\begin{theorem}
    The expectation and variance of frequency, persistence, and density of \(e\) after the \( i^{th} \) window is given by the following equation:
    \begin{equation}
    \begin{split}
        \mathbb{E}[f_i]&=i\lambda\\
        \mathbb{E}[p_i]&=i\cdot\left(1-e^{-\lambda}\right)\\
        \mathbb{E}[d_i]&=\frac{\lambda}{\left(1-e^{-\lambda}\right)}.
    \end{split}
    \end{equation}
        \begin{equation}
    \begin{split}
        \mathbb{VAR}[f_i]&=i\lambda\\
        \mathbb{VAR}[p_i]&=i\cdot e^{-\lambda}\cdot\left(1-e^{-\lambda}\right)\\
        \mathbb{VAR}[d_i]&\leq \frac{\lambda^2}{\left(1-e^{-\lambda}\right)^2}+\frac{\lambda+\lambda^2}{\left(1-e^{-\lambda}\right)}<\infty.
    \end{split}
    \end{equation}
\end{theorem}
\begin{proof}
    \begin{equation}
\mathbb{P}[n_i=k]=\frac{\lambda^k}{k!}e^{-\lambda}.
\end{equation}

\begin{equation}
    \mathbb{P}[n \geq 1] = 1 - \mathbb{P}[n_i=0] = 1 - e^{-\lambda}.
\end{equation}

In other words, the probability that flow \( e \) appears in each window is \( 1 - e^{-\lambda} \).

Given that there are \( i \) windows in total, the persistence \( p \) of the flow follows a binomial distribution:
\begin{equation}
    p \sim \text{Binomial}(i, 1 - e^{-\lambda}).
\end{equation}

\begin{equation}
    \mathbb{E}[f_i]=\mathbb{E}[\sum_{j=1}^{j=i}n_j]=\sum_{j=1}^{j=i}\mathbb{E}[n_j]=i\lambda.
\end{equation}

\begin{equation}
    \mathbb{E}[p_i]=\mathbb{E}[\sum_{j=1}^{j=i}I_{[n_j\geq 1]}]=\sum_{j=1}^{j=i}\mathbb{P}[n_j\geq 1]= i(1 - e^{-\lambda}).
\end{equation}

\begin{equation}
    \begin{split}
        \mathbb{E}[d_i]&=\mathbb{E}\left[\frac{\sum_{j=1}^{j=i}n_j}{\sum_{j=1}^{j=i}I_{[n_j\geq 1]}}\right]=\mathbb{E}\left[\mathbb{E}\left[\frac{\sum_{j=1}^{j=i}n_j}{\sum_{j=1}^{j=i}I_{[n_j\geq 1]}}|\sum_{j=1}^{j=i}I_{[n_j\geq 1]}\right]\right]\\
        &=\mathbb{E}\left[\mathbb{E}\left[\frac{\sum_{j=1}^{j=i}n_jI_{[n_j\geq 1]}}{\sum_{j=1}^{j=i}I_{[n_j\geq 1]}}|\sum_{j=1}^{j=i}I_{[n_j\geq 1]}\right]\right]=\frac{\lambda}{1 - e^{-\lambda}}.
    \end{split}
\end{equation}

    \begin{equation}
        \mathbb{VAR}[f_i]=\mathbb{VAR}[\sum_{j=1}^{j=i}n_j]=\sum_{j=1}^{j=i}\mathbb{VAR}[n_j]=i\lambda
    \end{equation}

    \begin{equation}
        \mathbb{VAR}[p_i]=\mathbb{VAR}[\sum_{j=1}^{j=i}I_{[n_j\geq 1]}]=\sum_{j=1}^{j=i}\mathbb{VAR}I_{[n_j\geq 1]}=i\cdot e^{-\lambda}\cdot\left(1-e^{-\lambda}\right)
    \end{equation}

    \begin{equation}
    \begin{split}
        \mathbb{VAR}[d_i]&=\mathbb{E}[d_i^2]-\mathbb{E}^2[d_i]=\mathbb{E}\left[ \left(\frac{\sum_{j=1}^{j=i}n_j}{\sum_{j=1}^{j=i}I_{[n_j\geq 1]}}\right)^2\right]-\frac{\lambda^2}{\left(1-e^{-\lambda}\right)^2}\\
        &=\mathbb{E}\left[\mathbb{E}\left[\left(\frac{\sum_{j=1}^{j=i}n_j}{\sum_{j=1}^{j=i}I_{[n_j\geq 1]}}\right)^2 |\sum_{j=1}^{j=i}I_{[n_j\geq 1]}\right]\right]-\frac{\lambda^2}{\left(1-e^{-\lambda}\right)^2}
    \end{split}
    \end{equation}

    \begin{equation}
        \begin{split}
            &\mathbb{E}\left[\left(\frac{\sum_{j=1}^{j=i}n_j}{\sum_{j=1}^{j=i}I_{[n_j\geq 1]}}\right)^2 |\sum_{j=1}^{j=i}I_{[n_j\geq 1]}\right]\\
            &=\left(\mathbb{E}[\sum_{j=1}^{j=i}n_j^2I_{[n_j\geq 1]}]+2\mathbb{E}[\sum_{j<k}]n_j n_k I_{[n_j\geq 1]} I_{[n_k\geq 1]} \right)/\left( \sum_{j=1}^{j=i}I_{[n_j\geq 1]}\right)^2\\
            &=\left( \frac{\lambda+\lambda^2}{1-e^{-\lambda}}\sum_{j=1}^{j=i}I_{[n_j\geq 1]}+\frac{\lambda^2}{\left(1-e^{-\lambda}\right)^2}\left(I_{[n_j\geq 1]}\right)\cdot\left(I_{[n_j\geq 1]}-1\right) \right)\\
            &/\left( \sum_{j=1}^{j=i}I_{[n_j\geq 1]}\right)^2\\
            &=\frac{\lambda^2}{\left(1-e^{-\lambda}\right)^2}+\left( \frac{\lambda+\lambda^2}{1-e^{-\lambda}}-\frac{\lambda^2}{\left(1-e^{-\lambda}\right)^2}\right)/\left( \sum_{j=1}^{j=i}I_{[n_j\geq 1]}\right)
        \end{split}
    \end{equation}

    \begin{equation}
    \begin{split}
        \mathbb{VAR}[d_i]
        &=\mathbb{E}\left[\frac{\lambda^2}{\left(1-e^{-\lambda}\right)^2}+\left( \frac{\lambda+\lambda^2}{1-e^{-\lambda}}-\frac{\lambda^2}{\left(1-e^{-\lambda}\right)^2}\right)/\left( \sum_{j=1}^{j=i}I_{[n_j\geq 1]}\right)\right]\\
        &=\frac{\lambda^2}{\left(1-e^{-\lambda}\right)^2}+\left( \frac{\lambda+\lambda^2}{1-e^{-\lambda}}-\frac{\lambda^2}{\left(1-e^{-\lambda}\right)^2}\right) \mathbb{E}\left[\frac{1}{\sum_{j=1}^{j=i}I_{[n_j\geq 1]}}\right]\\
        &\leq\frac{\lambda^2}{\left(1-e^{-\lambda}\right)^2}+\frac{\lambda+\lambda^2}{1-e^{-\lambda}}<\infty
    \end{split}
    \end{equation}

\end{proof}

According to \(L_p\) Convergence Theorem and 
Dominated Convergence Theorem, we can acknowledge that 
\begin{equation}
    \lim_{i\rightarrow \infty}\mathbb{E}[|d_i-\lambda|^2]=0,
\end{equation}
which demonstrates that \(d_i\) is \(L_2\) and almost surely converge to \(\lambda\).

Due to the constraint of \( p_0 \), the parameter \( \lambda \) of the $Poisson$ distribution in the PS flows under investigation should not be too small. However, considering the goal of identifying flows with low density, \( \lambda \) must also be kept within an appropriately limited range and should not be excessively large.

\subsection{Error Bound Analysis}

For PSSketch, the sources of error come from \textbf{Ejection Error}.

\subsubsection{Errors of Flow Ejection}

\begin{theorem}
 \( \hat{d_i} \) is an unbiased estimator of \( \lambda \), and \( \hat{d_i} \) \( L_2 \)-convergence as well as almost surely convergence to \( \lambda \). 

\end{theorem}
\begin{proof}
According to the structure of PSSketch \( \hat{d_i} \) represents the density of flow \( e \), counted from the last kicked window. Without loss of generality, we assume that the last time flow \( e \) was kicked occurred at the \( t \)-th window, denoted as event \( K_t \). Due to the stationary increment property of the \( d_i \) process, we can infer that
\begin{equation}
    \hat{d_i}|K_t=d_{i-t}\quad t=0, 1, 2, \ldots, i-1.
\end{equation}

\begin{equation}
    \mathbb{E}[\hat{d_i}]=\mathbb{E}[\mathbb{E}[\hat{d_i}|K_t]]=\mathbb{E}[\mathbb{E}[d_{i-t}|K_t]]=\frac{\lambda}{1 - e^{-\lambda}}.
\end{equation}

\begin{equation}
    \mathbb{VAR}[\hat{d_i}]=\mathbb{E}[\hat{d_i}^2]-\mathbb{E}^2[\hat{d_i}]=\mathbb{E}[\mathbb{E}[\hat{d_i}^2|K_t]]-\mathbb{E}^2[\hat{d_i}]<\infty
\end{equation}
    According to \(L_p\) Convergence Theorem and 
Dominated Convergence Theorem, we can acknowledge that 
\begin{equation}
    \lim_{i\rightarrow \infty}\mathbb{E}[|\hat{d_i}-\lambda|^2]=0,
\end{equation}
which demonstrates that \(\hat{d_i}\) is \(L_2\) and almost surely converge to \(\lambda\).
\end{proof}

\begin{theorem}
 \( \hat{d_i}-d_i \) is an unbiased estimator of \( 0 \), and \( \hat{d_i}-d_i \)  \( L_2 \)-convergence as well as almost surely convergence to \( 0 \). 
\end{theorem}

\begin{proof}
    \begin{equation}
        \mathbb{E}[\hat{d_i}-d_i]=\mathbb{E}[\mathbb{E}[\hat{d_i}|K_t]-d_i]=\mathbb{E}[\mathbb{E}[d_{i-t}]-d_i]=0
    \end{equation}
    \begin{equation}
        \mathbb{VAR}[\hat{d_i}-d_i]=\mathbb{E}[\left(\hat{d_i}-d_i\right)^2]=\mathbb{E}[\hat{d_i}^2]+\mathbb{E}[d_i^2]-2\lambda<\infty
    \end{equation}
    According to \(L_p\) Convergence Theorem and 
Dominated Convergence Theorem, we can acknowledge that 
    \begin{equation}
    \lim_{i\rightarrow \infty}\mathbb{E}[|\hat{d_i}-d_i|^2]=0,
\end{equation}
which demonstrates that \( \hat{d_i}-d_i\) is \(L_2\) and almost surely converge to \(0\).
\end{proof}

This demonstrates that, when the number of windows is sufficiently large, the error caused by ejection converges to zero. This means that even if a suspected PS flow is ejected, it will eventually be reported as a PS flow as long as the number of windows is large enough.

\subsection{Complexity analysis}

\subsubsection{Time Cost of PSSketch}

During the Insert phase, when a flow \( e \) is inserted into the competition layer, a hash function is applied once, with a time complexity of \( O(1) \). Similarly, when the flow is inserted into the protection layer, a hash function is applied once, with a time complexity of \( O(1) \). Therefore, the overall time complexity of the Insert phase is \( O(2) \).

\subsubsection{Space Cost of PSSketch}

\begin{theorem}
    The total storage space required by PSSketch is 
\begin{equation}
    XY \cdot (L_c^{FP} + L_c^{f} + L_c^{p} + L_c^{flag}) + R \cdot (L_c^{ID} + L_c^{f_{of}} + L_c^{p_{of}}),
\end{equation}
where \( R \ll X \cdot Y \).
\end{theorem}
\begin{proof}
    In the Competition Layer, there are \( X \) buckets, each containing \( Y \) cells. Each cell occupies a space of \( L_c^{FP} + L_c^{f} + L_c^{p} + L_c^{flag} \). In the Protection Layer, there is a single bucket containing \( R \) cells, and each cell occupies \( L_c^{ID} + L_c^{f_{of}} + L_c^{p_{of}} \). Therefore, the total storage space required is:
    \begin{equation}
        XY \cdot (L_c^{FP} + L_c^{f} + L_c^{p} + L_c^{flag}) + R \cdot (L_c^{ID} + L_c^{f_{of}} + L_c^{p_{of}}).
    \end{equation}

\end{proof}

\begin{theorem}
    The maximum frequency of a flow that can be stored in PSSketch is 
    \begin{equation}
        f_{max} = (2^{L_c^{f_{of}}} - 1) \cdot (2^{L_c^{f}} - 1) \approx 2^{L_c^{f_{of}} + L_c^{f}},
    \end{equation}

and the maximum persistence of a flow that can be stored is
\begin{equation}
    p_{max} = (2^{L_c^{p_{of}}} - 1) \cdot (2^{L_c^{p}} - 1) \approx 2^{L_c^{p_{of}} + L_c^{p}}.
\end{equation}
\end{theorem}

\begin{proof}
    According to the structure of PSSketch, the maximum frequency that can be stored is the product of the maximum count that the frequency counters in the Competition Layer can hold and the overflow count in the Protection Layer, i.e., 
    \begin{equation}
        f_{max} = (2^{L_c^{f_{of}}} - 1) \cdot (2^{L_c^{f}} - 1) \approx 2^{L_c^{f_{of}} + L_c^{f}}.
    \end{equation}
    Similarly, the maximum persistence that can be stored is
    \begin{equation}
        p_{max} = (2^{L_c^{p_{of}}} - 1) \cdot (2^{L_c^{p}} - 1) \approx 2^{L_c^{p_{of}} + L_c^{p}}.
    \end{equation}
\end{proof}

\begin{theorem}
    If PISketch stores flow data with the same order of magnitude for the maximum frequency and maximum persistence as PSSketch, then the Weight Sketch component of PISketch must have at least:
    \begin{equation}
        XY \cdot \left( L_c^{ID} + \log_2 L \cdot \left( L_c^{f_{of}} + L_c^{f} \right) + \left( L_c^{p_{of}} + L_c^{p} \right) \right)
    \end{equation}
\end{theorem}

\begin{proof}
    For the PISketch, the weight \( W \) must satisfy:
\begin{equation}
    2^{L_c^{W}} = L \cdot \left( L_c^{f_{of}} + L_c^{f} \right),
\end{equation}
which implies that \( L_c^{W} \) must be the smallest integer greater than or equal to \( \log_2 L \cdot \left( L_c^{f_{of}} + L_c^{f} \right) \). Similarly, to store \( p_{max} \approx 2^{L_c^{p_{of}} + L_c^{p}} \), the parameter \( L_c^{N} \) in PISketch must be at least \( L_c^{p_{of}} + L_c^{p} \). Therefore, the storage space required for the Weight Sketch in PISketch is greater than or equal to:
\begin{equation}
    XY \cdot \left( L_c^{ID} + \log_2 L \cdot \left( L_c^{f_{of}} + L_c^{f} \right) + \left( L_c^{p_{of}} + L_c^{p} \right) \right).
\end{equation}
\end{proof}

\begin{theorem}
    If the Weight Sketch in PISketch and the Competition Layer in PSSketch occupy storage space of the same order of magnitude, then:
    \begin{equation}
        p_{max}^{PISketch} = 2^{L_c^{N}} - 1 \ll p_{max}.
    \end{equation}
\end{theorem}

\begin{proof}
    In this case, we have:
    \begin{equation}
        p_{max}^{PISketch} = 2^{L_c^{N}} - 1 = 2^{L_c^{p}} - 1 \ll (2^{L_c^{p_{of}}} - 1) \cdot (2^{L_c^{p}} - 1) = p_{max}.
    \end{equation}
\end{proof}

Overall, the above arguments demonstrate that, compared to PISketch, PSSketch has a significant advantage in terms of space complexity due to its handling of overflows and the use of a dual-layer mechanism consisting of the Competition Layer and Protection Layer.


